\documentclass[envcountsame]{article}

\usepackage{a4wide}
\usepackage[utf8x]{inputenc}
\usepackage{amsmath,amssymb,amsthm}
\usepackage{amsfonts}
\usepackage[english]{babel}
\usepackage{fontenc}
\usepackage{graphicx}
\usepackage{tikz}
\usepackage{xspace}
\usepackage{enumerate,a4wide}
\usepackage{subfigure}
\usepackage[pdfborder={0 0 0}]{hyperref}
\usepackage{calc}

\usepackage{authblk}
\usepackage{color}
\usetikzlibrary{calc,arrows,decorations.pathreplacing,positioning,backgrounds,calc,trees}

\newtheorem{theorem}{Theorem}

\newtheorem{definition}[theorem]{Definition}
\newtheorem{lemma}[theorem]{Lemma}
\newtheorem{corollary}[theorem]{Corollary}

\newtheorem{example}[theorem]{Example}

\makeatletter
\@ifpackagelater{tikz}{2013/12/01}{%
  \def\AtanTwo(#1,#2){atan2({#1},{#2})}%
}{%
  \def\AtanTwo(#1,#2){atan2({#2},{#1})}%
}
\makeatother

\newcommand{\convexpath}[2]{
  [   
  create hullcoords/.code={
    \global\edef\namelist{#1}
    \foreach [count=\counter] \nodename in \namelist {
      \global\edef\numberofnodes{\counter}
      \coordinate (hullcoord\counter) at (\nodename);
    }
    \coordinate (hullcoord0) at (hullcoord\numberofnodes);
    \pgfmathtruncatemacro\lastnumber{\numberofnodes+1}
    \coordinate (hullcoord\lastnumber) at (hullcoord1);
  },
  create hullcoords
  ]
  ($(hullcoord1)!#2!-90:(hullcoord0)$)
  \foreach [
  evaluate=\currentnode as \previousnode using \currentnode-1,
  evaluate=\currentnode as \nextnode using \currentnode+1
  ] \currentnode in {1,...,\numberofnodes} {
    let \p1 = ($(hullcoord\currentnode) - (hullcoord\previousnode)$),
    \n1 = {\AtanTwo(\y1,\x1) + 90},
    \p2 = ($(hullcoord\nextnode) - (hullcoord\currentnode)$),
    \n2 = {\AtanTwo(\y2,\x2) + 90},
    \n{delta} = {Mod(\n2-\n1,360) - 360}
    in 
    {arc [start angle=\n1, delta angle=\n{delta}, radius=#2]}
    -- ($(hullcoord\nextnode)!#2!-90:(hullcoord\currentnode)$) 
  }
}

\newcommand{\ns}{{\mathbb N}} 
\newcommand{\zs}{{\mathbb Z}} 

\newcommand{\abs}[1]{\mathsf{abs}({#1})}

\newcommand{\gaps}[1]{\mathrm{gaps}({#1})}

\newcommand{\mc}{\mathcal}
\newcommand{\bb}{\mathbb}
\newcommand{\val}{\mathrm{val}}
\newcommand{\mcT}{\mathcal T(\mathcal F)}
\newcommand{\rk}{\mathrm{rank}}

\renewcommand{\P}{{\sf P}\xspace}
\newcommand{\NP}{{\sf NP}\xspace}

\newcommand{\PSPACE}{{\sf PSPACE}\xspace}
\newcommand{\LogCFL}{{\sf LogCFL}\xspace}

\begin{document}

\title{Tree Compression Using String Grammars\thanks{The third and fourth author are supported by the DFG-project LO 748/10-1 (QUANT-KOMP).}}

\author{Moses Ganardi, Danny Hucke, Markus Lohrey and Eric Noeth}
\affil{Universit\"at Siegen, Germany\\
\texttt{\{ganardi,hucke,lohrey,eric.noeth\}@eti.uni-siegen.de}}


\date{} 

\maketitle

\begin{abstract}
We study the compressed representation of a ranked tree by a (string) straight-line program (SLP) for its 
preorder traversal, and compare it with the
well-studied representation by straight-line context free tree grammars (which are also known as tree straight-line programs or TSLPs).
Although SLPs turn out to be exponentially more succinct than TSLPs, we show that
many simple tree queries can still be performed efficiently on SLPs, such as computing
the height and Horton-Strahler number of a tree, tree navigation, or evaluation of Boolean expressions.
Other problems on tree traversals turn out to be intractable, e.g. pattern matching and
evaluation of tree automata.
\end{abstract}
 
\section{Introduction}

{\em Grammar-based compression} has become an active field in string compression during the past 20 years.
The idea is to represent a given string $s$ by a small context-free grammar that generates only $s$;
such a grammar is also called a {\em straight-line program} (SLP). For instance, the word $(ab)^{1024}$ can be represented
by the SLP with the productions $A_0 \to ab$ and $A_i \to A_{i-1} A_{i-1}$ for $1 \leq i \leq 10$ ($A_{10}$ is the start 
symbol). The size of this SLP (the size of an SLP is usually defined as the total length of all right-hand sides of the productions)
is much smaller than the length of the string $(ab)^{1024}$. In general, an SLP of size $n$ 
can produce a string of length $2^{\Omega(n)}$. Hence, an SLP can be seen indeed as a succinct representation
of the generated string. The goal of grammar-based string compression is to construct from a given input
string $s$ a small SLP that produces $s$. Several algorithms for this have been proposed and analyzed. 
Prominent grammar-based string compressors are for instance 
{\sf LZ78}, {\sf RePair}, and {\sf BISECTION}, see \cite{CLLLPPSS05} for more details.
The theoretically  best known  polynomial time 
grammar-based compressors  \cite{CLLLPPSS05,Jez13approx,Ryt03,Sakamoto05}  approximate the size of a smallest SLP up to a factor 
$\mathcal{O}(\log (n/g))$, where $g$ is the size of a smallest
SLP for the input string. 

Motivated by applications where large tree structured data occur, like XML processing,
gram\-mar-based compression has been extended to trees \cite{MLMN13,BuLoMa07,JezLo14approx,LohreyMM13}, see \cite{Loh15} for  a survey. 
Unless otherwise specified, a tree in this paper is always a rooted ordered tree over a ranked alphabet, 
i.e., every node is labelled with a symbol and the rank of this symbol is equal to the number of children of the node.
This class of trees occurs in many different contexts like for instance term rewriting, expression evaluation, tree automata, and
functional programming. A tree over a ranked alphabet is uniquely represented by its preorder traversal string.
For instance, the preorder traversal of the tree $f(g(a),f(a,b))$ is the string $fgafab$.
It is now a natural idea to apply a string compressor to this preorder traversal.
In this paper we study the compression of ranked trees by SLPs for their preorder traversals.
This approach is very similar to \cite{BLRSSW15}, where unranked unlabelled trees are compressed by SLPs for their
balanced parenthesis representations. In \cite{NaOP14} this idea is used together with the grammar-based compressor
RePair to get a new compressed suffix tree implementation.

In Section~\ref{section:slpvstslp} we compare the size of SLPs for preorder traversals with two other grammar-based
compressed tree representations: the above mentioned SLPs for balanced parenthesis representations from 
\cite{BLRSSW15} and (ii)  tree straight-line programs (TSLPs) \cite{BuLoMa07,HuckeLN14,JezLo14approx,LohreyMM13}. The latter 
directly generalize string SLPs to trees using context-free tree grammars
that produce a single tree, see \cite{Loh15} for a survey.  
TSLPs generalize dags (directed acyclic graphs), which are widely used as a compact tree representation. Whereas
dags only allow to share repeated subtrees, TSLPs can also share repeated internal tree patterns.
In  \cite{HuckeLN14arxiv} it is shown that every tree of size $n$ over a fixed ranked alphabet can be produced
by a TSLP of size $\mc O(\frac{n}{\log n})$,
which is worst-case optimal. A grammar-based tree compressor based on TSLPs with 
an  approximation ratio of $\mc O(\log n)$ was presented in \cite{JezLo14approx}.
In \cite{BuLoMa07}, it was shown that from a given
TSLP $\bb A$ of size $m$ for a tree $t$ one can efficiently construct an SLP of size $\mc O(m \cdot r)$ for the preorder traversal of $t$, 
where $r$ is the maximal rank occurring in $t$
(i.e., the maximal number of children of a node).
Hence, a smallest SLP for the traversal of $t$ cannot be much larger than a smallest TSLP for $t$. 
Our first main result (Theorem~\ref{theorem:SmallestTravGrammar}) shows that SLPs can be exponentially more succinct than TSLPs: 
We construct a family of binary trees $t_n$ ($n \geq 0$) such that the size of a smallest SLP for the traversal of 
$t_n$ is polynomial in $n$ but the size of a smallest TSLP for $t_n$ is 
$\Omega(2^{n/2})$.  We also match this lower bound by an upper bound: 
Given an SLP $\bb A$ of size $m$ for the traversal of a tree $t$ of height $h$ and maximal rank $r$,
one can efficiently construct a TSLP for $t$ of size $\mc O(m \cdot h \cdot r)$ (Theorem~\ref{theorem:TraversalToTree}).
Finally, we construct a family of binary trees $t_n$ ($n \geq 0$) such that the size of a smallest SLP for the 
preorder traversal of $t_n$ is polynomial in $n$ but the size of a smallest SLP for the balanced parenthesis representation is
$\Omega(2^{n/2})$ (Theorem~\ref{theorem:BP}). Hence, SLPs for preorder traversals can be exponentially more succinct
than SLPs for balanced parenthesis representations.  It remains open, whether the opposite behavior is possible as well.

We also study algorithmic problems for trees that are encoded by SLPs. 
We extend some of the results from \cite{BLRSSW15} on querying SLP-compressed
balanced parenthesis representations to our context. Specifically, we show that
after a linear time preprocessing we can navigate (i.e., move to the parent node and the $k^{\text{th}}$ child),
compute lowest common ancestors and subtree sizes in time $\mc O(\log N)$, where $N$ is the size of the 
tree represented by the SLP (Theorem~\ref{thm:navi}). For a couple of other problems (computation of the height and depth of a node,
computation of the Horton-Strahler number, 
and evaluation of Boolean expressions) we provide at least polynomial time algorithms for the case that the input
tree is given by an SLP for the preorder traversal.
On the other hand, there exist problems that are polynomial time solvable for TSLP-compressed
trees but difficult for SLP-compressed trees: Examples for such problems are  pattern matching, evaluation of max-plus expressions, and
membership for tree automata. 
Looking at tree automata is also interesting when compared with the situation for 
explicitly given (i.e., uncompressed) preorder traversals. For these, evaluating Boolean expressions (which is the membership
problem for a particular tree automaton) is $\mathsf{NC}^1$-complete by a famous result of Buss \cite{Bus87}, and the $\mathsf{NC}^1$ upper bound was generalized 
to every fixed tree automaton  \cite{loh01}. If we compress the preorder traversal by an SLP, the problem is still solvable in polynomial time for 
Boolean expressions (Theorem~\ref{theorem:eval_bool}), but there is a fixed tree automaton where the evaluation problem becomes $\PSPACE$-complete 
(Theorem~\ref{thm-TA}).

\paragraph{\bf Related work on tree compression.}

There are also  tree compressors based on other grammar formalisms.
In \cite{Akutsu10} so called elementary ordered tree grammars are used, and a polynomial time compressor
with an approximation ratio of $\mc O(n^{5/6})$ is presented.
Also the {\em top dags} from \cite{BilleGLW13} can be
seen as a variation of TSLPs for unranked trees. Recently, in \cite{Hubschle-Schneider15} it was shown that for every tree of size $n$
with $\sigma$ many node labels,
the top dag has size $\mc O(\frac{n \cdot \log \log_\sigma n}{\log_\sigma n})$, which improved the bound from \cite{BilleGLW13}.
An extension of TSLPs to higher order tree grammars was 
proposed in \cite{KobayashiMS12}.

Another class of tree compressors use succinct data structures for trees. Here, the goal
is to represent a tree in a number of bits that asymptotically matches the information
theoretic lower bound, and at the same time allows efficient querying (ideally in time $\mc O(1)$) of the data structure.
For unlabelled unranked trees of size $n$ there exist representations with $2n+o(n)$ bits that support navigation and some other tree queries
in time $\mc O(1)$ \cite{BenoitDMRRR05,Jacobson89,JanssonSS12,MunroR01}. This result has been extended to labelled trees, where $(\log \sigma) \cdot n + 2n + o(n)$ 
bits suffice when $\sigma$ is the number of node labels \cite{FerraginaLMM09}.

\section{Preliminaries}

Let $\Sigma$ be a finite alphabet. 
For a string $w = a_1 \cdots a_n \in \Sigma^*$ we define $|w|=n$, $w[i] = a_i$ and $w[i : j] = a_i \cdots a_j$ 
where $w[i : j] = \varepsilon$, if $i > j$. Let $w[:i] = w[1:i]$ and $w[i:] = w[i:|w|]$.
With $\mathrm{rev}(w) = a_n \cdots a_1$ we denote $w$ reversed.
Given two strings $u, v \in \Sigma^*$, the {\em convolution}
$u \otimes v \in (\Sigma \times \Sigma)^*$ 
is the string of length $\min\{|u|,|v|\}$ defined by $(u \otimes v)[i] = (u[i],v[i])$ for $1 \le i \le \min\{|u|,|v|\}$.

\subsection{Complexity classes}

  We assume familiarity with the basic classes from complexity theory, in particular 
  \P, \NP and \PSPACE.  The following definitions are only needed in Section~\ref{sec-arithmetic-eval}.
  The counting class $\# \P$ contains all functions $f : \Sigma^* \to \ns$ for which there exists a nondeterministic
  polynomial time machine $M$ such that for every $x \in \Sigma^*$, $f(x)$ is the number of accepting computation
  paths of $M$ on input $x$. The class $\mathsf{PP}$ (probabilistic polynomial time) contains all problems $A$ for which there exists
  a nondeterministic
  polynomial time machine $M$ such that for every input $x$: $x \in A$ if and only if more than half of all computation
  paths of $M$ on input $x$ are accepting.
  By a famous result of Toda \cite{To91}, the class $\P^{\mathsf{PP}} = \P^{\#\P}$ (i.e., the class of all languages
  that can be decided in deterministic polynomial time with the help of an oracle from $\mathsf{PP}$  contains
  the whole polynomial time hierarchy. Hence, if a problem is $\mathsf{PP}$-hard, then this can be seen
  as a strong indication that the problem does not belong to the polynomial time hierarchy (otherwise the 
  polynomial time hierarchy would collapse).
  
   The levels of the {\em counting hierarchy} $\mathsf{C}^p_i$ ($i \geq 0$) are inductively
   defined as follows: $\mathsf{C}^p_0 = \P$ and  $\mathsf{C}^p_{i+1} = \mathsf{PP}^{\mathsf{C}^p_{i}}$
   (the set of languages accepted by a $\mathsf{PP}$-machine as above with an oracle from $\mathsf{C}^p_{i}$) 
   for all $i \geq 0$. Let $\mathsf{CH} = \bigcup_{i\geq 0} \mathsf{C}^p_{i}$ be the counting hierarchy. It is not 
   difficult to show that $\mathsf{CH}  \subseteq \PSPACE$, and most complexity theorists conjecture that
   $\mathsf{CH}  \subsetneq \PSPACE$. Hence, if a problem belongs to the counting hierarchy, then this can be seen
  as an indication that the problem is probably not $\PSPACE$-complete. 
  The counting hierarchy can be also seen as an exponentially blown-up version of the circuit complexity class
  $\mathsf{DLOGTIME}$-uniform $\mathsf{TC}^0$. This is the class of all languages that can be decided
  with a constant-depth polynomial-size circuit family of unbounded fan-in that in addition to normal Boolean
  gates may also use threshold gates. $\mathsf{DLOGTIME}$-uniformity means that one can compute in time
  $\mc O(\log n)$ (i) the type of a given gate of  the $n^{\text{th}}$ circuit, and (ii)
  whether two given gates of the $n^{\text{th}}$ circuit are connected by a wire. Here, gates of the 
  $n^{\text{th}}$ circuit  are encoded by bit string of length $\mc O(\log n)$.
   More details on the counting hierarchy (resp., circuit complexity) can be found in \cite{DBLP:journals/eatcs/AllenderW90} (resp.,
   \cite{Vol99}).

\subsection{Trees}

	A \emph{ranked alphabet} $\mc F$ is a finite set of symbols where every symbol $f \in \mc F$ 
	has a rank $\rk(f) \in \ns$. 
 	We assume that $\mc F$ contains at least one symbol of rank zero. 	
	By $\mc F_n$ we denote the symbols of $\mc F $ of rank $n$. 
	Later we will also allow
	ranked alphabets, where $\mc F_0$ is infinite.
	For the purpose of this paper,
	it is convenient to define trees as particular strings over the alphabet $\mc F$ (namely as preorder traversals).
	The set $\mcT$ of all \emph{trees} over $\mc F$ is the subset of $\mc F^*$ defined inductively
	as follows: If $f \in \mc F_n$ with $n \geq 0$ and $t_1, \ldots, t_n \in \mcT$, then also
	$f  t_1 \cdots t_n \in \mcT$.
	
	We call a string $s \in {\mc F}^*$ a \emph{fragment}
	if there exists a tree $t \in \mc T (\mc F)$ and a non-empty string $x\in{\mc F}^+$ such that $sx=t$.
	Note that the empty string $\varepsilon$ is a fragment.
	Intuitively, a fragment is a tree with gaps.
	The number of gaps of a fragment $s \in {\mc F}^+ $ is formally defined as
	the number $n$ of trees $t_1, \dots, t_n \in \mcT$ such that
	$s t_1 \cdots t_n \in \mcT$, and is denoted by $\gaps{s}$.
	The number of gaps of the empty string is defined as $0$.
	The following lemma states that $\gaps{s}$ is indeed well-defined.
		
	\begin{lemma}\label{lemma:traversal-folklore}
	The following statements hold:
	 \begin{itemize}
	  \item The set $\mcT$ is prefix-free, i.e. $t \in \mcT$ and $tv \in \mcT$ imply $v = \varepsilon$.
	  \item If $t \in \mcT$, then every suffix of $t$ factors uniquely into
        a concatenation of strings from $\mcT$.
      \item For every fragment $s \in {\mc F}^+$ there is a unique $n \geq 1$
      such that
     $\{ x \in \mc F^* \mid sx \in \mcT \} = (\mcT)^n$.
	 \end{itemize}
	\end{lemma}
	Since $\mcT $ is prefix-free we immediately get:
	
 \begin{lemma}\label{lemma:substring_traversal_string}
 For every $w \in {\mc F}^*$ there exist unique $n\ge 0$, $t_1,\dots,t_{n} \in \mcT$ and a unique fragment $s$ such that
 $w = t_1 \cdots t_{n} s$.
 \end{lemma}
 Let $w \in {\mc F}^*$ and let $w = t_1 \cdots t_n s$ as in Lemma~\ref{lemma:substring_traversal_string}.
 We define 
 $c(w) = (n,\gaps{s})$.
 The number $n$ counts the number of full trees in $w$ and
 $\gaps{s}$ is the number of trees 
 missing to make the fragment $s$ a tree, too.
 	
For better readability, we occasionally write a tree $f t_1 \cdots t_n$
with $f \in \mc F_n$ and $t_1, \ldots, t_n \in \mcT$ as $f(t_1, \ldots, t_n)$, which corresponds
to the standard term representation of trees. 
We also consider trees in their graph-theoretic interpretation 
where the set of nodes of a tree $t$ is the set of positions $\{1, \ldots, |t|\}$ of the string $t$.
The root node is $1$. If $t$ factorizes as $u f t_1 \cdots t_n v$ for $u, v \in \mc F^*$,
$f \in \mc F_n$, and $t_1, \ldots, t_n \in \mcT$, then the $n$ children of node $|u|+1$ are
$|u|+2 + \sum_{i=1}^k |t_i|$ for $0 \leq k \leq n-1$. 
We define the depth of a node in $t$ (number of edges from the root to the node)  and
 the height of $t$ (maximal depth of a node) as usual.
 Note that the tree $t$ as a string is simply the preorder traversal of the tree $t$ seen in its standard
 graph-theoretic interpretation.

    \begin{example}\label{ex:small_tree}
     Let $t = ffaafffaaaa =  f(f(a,a),f(f(f(a,a),a),a))$ be the tree depicted in Figure~\ref{fig:example} with
     $f\in\mc F_2 $ and $a\in\mc F_0$. Its height is $4$.
     All prefixes (including the empty word, excluding the full word) of $t$ are fragments.
	 The fragment $s=ffaafff$
	 is also depicted in Figure~\ref{fig:example} in a graphical way. The dashed edges visualize the gaps.
	 We have $\gaps{s} = 4$. For the factor $u = aafffa$ of $t$ we have $c(u)=(2,3)$. The children of node $5$
	 (the third $f$-labelled node) are $6$ and $11$.
    \end{example}
\begin{figure}[t]
\centering

\begin{minipage}[hbt]{0.4\textwidth}
\centering

\begin{tikzpicture}[-,scale=1,level distance=5mm]
 \tikzset{level 1/.style={sibling distance=25mm}}
 \tikzset{level 2/.style={sibling distance=13mm}}

\node (root){$f$}
  child {node (0) {$f$}
    child {node (00) {$a$}}
    child {node (01) {$a$}} 
  }
  child {node (1) {$f$}
    child {node (10) {$f$}
      child {node (100) {$f$}
        child {node (1001) {$a$}}
        child {node (1001) {$a$}}        
      }
      child {node (101) {$a$}}
    }
    child {node (11) {$a$}}
  }
;
\end{tikzpicture}
\end{minipage}
\hfill
\begin{minipage}[hbt]{0.4\textwidth}
\centering
\begin{tikzpicture}[-,scale=1,level distance=5mm]
 \tikzset{level 1/.style={sibling distance=25mm}}
 \tikzset{level 2/.style={sibling distance=13mm}}
 
\node (root){$f$}
  child {node (0) {$f$}
    child {node (00) {$a$}}
    child {node (01) {$a$}} 
  }
  child {node (1) {$f$}
    child {node (10) {$f$}
      child {node (100) {$f$}
        child[dashed] {node (1001) {}}
        child[dashed] {node (1001) {}}        
      }
      child[dashed] {node (101) {}}
    }
    child[dashed] {node (11) {}}
  }
;
\end{tikzpicture} 
\end{minipage}
\caption{The tree $t$ from Example~\ref{ex:small_tree} 
and the tree fragment corresponding to the fragment 
  $ffaafff$.}
 \label{fig:example}
\end{figure}
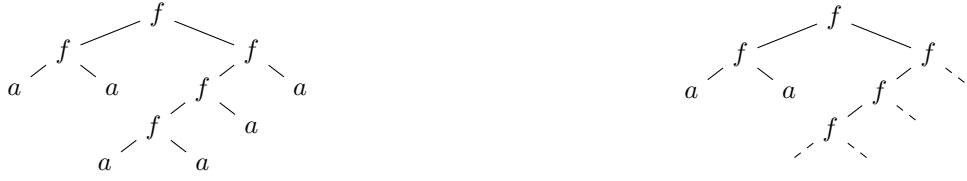

\subsection{Straight-line programs}

 A \emph{straight-line program}, briefly SLP, is a context-free grammar that
 produces a single string. Formally, it is a tuple $\bb A = (N,\Sigma, P, S)$, where $N$ is a 
 finite set of nonterminals, $\Sigma$ is a finite set of terminal symbols 
 ($\Sigma \,\cap \, N = \emptyset$), 
 $S \in  N$ is the start nonterminal, and $P$ is a finite
 set of productions (or rules) of the form $A \to w$ for $A \in N$, $w \in (N \cup \Sigma)^*$ such that:
 \begin{itemize}
  \item For every $A \in N$, there exists exactly one production of the form $A \to w$, and
  \item the binary relation $\{ (A, B) \in N \times N \mid (A \to w) \in P,\;B \text{ occurs in } w \}$
       is acyclic.
 \end{itemize} 
 Every nonterminal $A \in N$ produces a unique string $\val_{\bb A}(A) \in \Sigma^*$.
 The string defined by $\bb A$ is $\val(\bb A) = \val_{\bb A}(S)$. 
 We usually omit the subscript $\bb A$ when the context is clear.
 The \emph{size} of the SLP $\bb A$ is 
 $|\bb A| = \sum_{(A \to w) \in P} |w|$.
 One can transform an SLP $\bb A = (N,\Sigma,P,S)$ which produces a nonempty word in linear time into
 {\em Chomsky normal form}, i.e. for each production $(A \to w) \in P$, either $w \in \Sigma$
 or $w = BC$ where $B,C \in N$. 
 
 For an SLP $\bb A$ of size $n$ we have $|\val(\bb A)| \in 2^{\mc O(n)}$, and there exists a family of SLPs $\bb A_n$ ($n \geq 1$)
 such that $|\bb A_n| \in \mathcal{O}(n)$ and $|\val(\bb A)| = 2^n$. Hence, SLPs allow exponential compression.
 
 The following lemma summarizes known results about SLPs which we will use throughout the paper, see e.g. \cite{LoCWP}.
 
 \begin{lemma}\label{lemma:slp_folklore}
	There are linear time algorithms for the following problems:
	\begin{enumerate}
	        \item \label{lemma:occ_term}
		Given an SLP $\bb A$, compute the set of symbols occurring in $\val(\bb A)$.
		\item \label{lemma:number-of-occ}
		Given an SLP $\bb A$ with terminal alphabet $\Sigma$ and a subset $\Gamma \subseteq \Sigma$, 
		compute the number of occurrences of symbols from $\Gamma$ in $\val(\bb A)$.
	        \item \label{lemma:firstpos}
		Given an SLP $\bb A$ with terminal alphabet $\Sigma$, a subset $\Gamma \subseteq \Sigma$, and
		a number $i$, compute the position of the $i^{\text{th}}$ occurrence of a symbol from $\Gamma$ in 
		$\val(\bb A)$ (if it exists).
		\item \label{lemma:cut}
		Given an SLP $\bb A$ and  $i,j \in \{1, \dots, |\val(\bb A)|\}$ where $i \le j$, compute an SLP for $\val(\bb A) [i:j]$. The 
		size of the SLP for $\val(\bb A) [i:j]$ is bounded by $\mc O(|\bb A|)$.
  \end{enumerate}
 \end{lemma}

\subsection{Tree straight-line programs}

 We now define tree straight-line programs.
 Let $\mc F $ and $\mc V$ be two disjoint ranked alphabets,
 where we call elements from $\mc F$ \emph{terminals} and 
 elements from $\mc V$ \emph{nonterminals}.
 Let further $\mc X=\left\{x_1,x_2,\dots\right\}$ be a countably infinite set of \emph{parameters} (disjoint from $\mc F$ and $\mc V$),
 which we treat as symbols of rank zero.
 In the following we consider trees over $\mc F \cup \mc V \cup \mc X$.
 The \emph{size} $|t|$ of such a tree $t$ is defined as the number of nodes
 labelled by a symbol from $\mc F \cup\mc V$, i.e. we do not count parameter nodes.
 A \emph{tree straight-line program} $\bb A$, or short \emph{TSLP}, is a tuple $\bb A
 =(\mc V, \mc F, P, S)$, where $\mc V$ is the set of nonterminals,
 $\mc F$ is the set of terminals,
 $S \in \mc V_0$ is the start nonterminal and
 $P$ is a finite set of productions of the form
 $A(x_1, \dots, x_n) \to t$
 (which is also briefly written as $A \to t$),
 where $n \ge 0$, $A \in \mc V_n$ and 
 $t \in \mc T (\mc F \cup \mc V \cup \mc \{x_1,\dots, x_n\})$ 
 is a tree in which every parameter $x_i$ ($1 \leq i \leq n$) occurs at most once, such that:
 \begin{itemize}
   \item For every $A \in \mc V_n$ there exists exactly one production of the form
    $A(x_1, \dots, x_n) \to t$, and
   \item the binary relation $\{ (A,B) \in \mc V \times \mc V \mid  (A \to t) \in P, B \text{ is a label in }t\}$ is acyclic.
 \end{itemize}
 These conditions ensure that exactly one tree 
 $\val_{\bb A}(A) \in \mc T(\mc F\cup\{x_1, \ldots, x_n\})$ 
 is derived from every nonterminal $A \in \mc V_n$ 
 by using the rules as rewriting rules in the usual sense. 
 As for SLPs, we omit the subscript $\bb A$ when the context is clear.
 The tree defined by $\bb A$ is $\val(\bb A) = \val_{\bb A}(S)$.
  The \emph{size} $|\bb A|$ of a TSLP $\bb A=(\mc V,\mc F,P,S)$ 
  is $|\bb A|=\sum_{(A\to t)\in P}|t|$.
  We call a TSLP \emph{monadic} if every nonterminal has rank at most one.
  One can transform every TSLP $\bb A$ into a monadic one of size $\mc O(|\bb A| \cdot r)$, 
  where $r$ is the maximal rank of a terminal in $\bb A$~\cite{LoMaSS12}. TSLPs, where every
  nonterminal has rank $0$ correspond to dags (the nodes of the dag are the nonterminals of the TSLP).
  
  For a TSLP $\bb A$ of size $n$ we have $|\val(\bb A)| \in 2^{\mc O(n)}$, and there exists a family of TSLPs $\bb A_n$ ($n \geq 1$)
 such that $|\bb A_n| \in \mathcal{O}(n)$ and $|\val(\bb A)| = 2^n$. Hence, analogously to SLPs, TSLPs allow exponential compression.
  One can also define {\em nonlinear} TSLPs where parameters can occur multiple times on right-hand sides;
  these can achieve doubly exponential compression
  but have the disadvantage that many algorithmic problems become more difficult, see e.g. \cite{DBLP:journals/tcs/LohreyM06}.

  For every word $w$ (resp., tree $t$) there exists a smallest SLP (resp., TSLP) $\bb A$. 
  It is known that, unless $\mathsf{P=NP}$, 
  there is no polynomial time algorithm that finds a smallest SLP (resp., TSLP) for a given 
  word~\cite{CLLLPPSS05} (resp. tree).
  
\section{Checking whether an SLP produces a tree}

 In this section we show that, given an SLP $\bb A$ and a ranked alphabet $\mc F$,
 we can verify in time linear in $|\bb A|$, whether $\val(\bb A) \in \mcT$.
 In other words, we present a linear time algorithm for the compressed membership problem 
 for the language $\mcT \subseteq \mc F^*$.
 We remark that $\mcT$ is a context-free language,
 which can be seen by considering the grammar with productions 
 $S \to fS^n$ for all symbols $f \in \mc F_n$. 
 In general the compressed membership problem for context-free languages can be solved in \PSPACE and
 there exists a deterministic context-free language with a \PSPACE-complete  compressed membership problem 
 \cite{DBLP:journals/jcss/CaussinusMTV98,DBLP:journals/iandc/Lohrey11}.

 \begin{theorem} \label{thm-rec-trav}
  Given an SLP $\bb A$, 
  one can check in time $\mc O(\bb A)$, whether
  $\val(\bb A) \in \mcT$.
 \end{theorem}
 
  \begin{proof}
  Let $\bb A=(N,{\mc F}, P, S)$ be in Chomsky normal form and let $A \in N$.
  Due to Lemma~\ref{lemma:substring_traversal_string}, we know that 
  $\val(A)$ is the concatenation of trees and a (possibly empty) fragment.
  Define $c(A) := c(\val(A))$. Then $\val(\bb A) \in \mcT$ if and only if  $c(S) = (1,0)$.
  Hence, it suffices to compute  $c(A)$  for all nonterminals $A \in N$. We do this bottom-up.
  If $(A \to f) \in P$ with $f \in \mc F_n$, then we have
	\[
		c(A) = \begin{cases}
			(1, 0) & \text{if } n = 0 \\
			(0, n) & \text{otherwise}.
		\end{cases}
	\]
Now consider a nonterminal $A$ with the rule 	
 $(A\to BC)\in P$, and let $c(B)=(b_1,b_2)$, $c(C)=(c_1,c_2)$. We claim that
	\[
		c(A) = \begin{cases}
			(b_1 + c_1 - \max\{1,b_2\} + 1, c_2) & \text{if } b_2 \le c_1 \\
			(b_1, c_2 + b_2 - c_1 - \min\{1,c_2\}) & \text{otherwise.}
		\end{cases}
	\]
   Let $\val(B) = t_1 \cdots t_{b_1} s$ and
    $\val(C) = t'_1 \cdots t'_{c_1} s'$, where $t_1, \ldots, t_{b_1},  t'_1, \ldots, t'_{c_1} \in \mcT$
    and $s$ (resp., $s'$) is a fragment with $\gaps{s} = b_2$ (resp., $\gaps{s'} = c_2$). We distinguish two cases:  
    
    \medskip
    \noindent
    {\em Case} $b_2 \le c_1$: If $b_2 \geq 1$, then the string
    $s t'_1 \cdots t'_{b_2}$
    is a tree, and thus $\val(A)$
    contains $b_1 + 1 + (c_1-b_2)$ full trees and the fragment $s'$ with $c_2$ many gaps.
    On the other hand, if $b_2= 0$, then $\val(A)$ contains $b_1 + c_1 $ many full trees.
    
     \medskip
    \noindent
    {\em Case}  $b_2 > c_1$:
    The trees  $t'_1,\dots, t'_{c_1}$ fill $c_1$ many gaps of $s$, and if $s'\ne\varepsilon$, then the fragment $s'$
    fills one more gap, while creating another $c_2$ gaps. In total there are
    $b_2-(c_1+1)+c_2$ gaps if $c_2 > 0$ and 
    $b_2-c_1$ gaps if $c_2=0$. 
 \end{proof}

\section{SLPs for traversals versus other grammar-based tree representations}
\label{section:slpvstslp}

In this section, we compare the worst case size of SLPs for traversals with the following
two grammar-based tree representations:
\begin{itemize}
\item TSLPs, and
\item SLPs for balanced parenthesis sequences \cite{BLRSSW15}.
\end{itemize}

\subsection{SLPs for traversals versus TSLPs}

In \cite{BuLoMa07} it is shown that a TSLP $\bb A$ producing a tree $t\in \mcT$
can always be transformed into an SLP of size $\mc O(|\bb A| \cdot r)$ producing $t$,
where $r$ is the maximal rank of a label occurring in $t$.
So, for binary trees the size at most doubles.
In this section we will discuss the other direction, i.e. transforming an SLP into a TSLP.
Let $a$ be a symbol of rank $0$ and let $f_n$ be a symbol of rank $n$ for each $n \in \ns$.
Now let $t_n$ be the tree $f_n a^n$ and consider the family of trees
$(t_n)_{n\in\ns}$ with unbounded rank.
The size of the smallest TSLP for $t_n$ is $n+1$, whereas the size of the smallest SLP for $t_n$ is in $\mc O (\log n)$.
It is less obvious that such an exponential gap can be also realized with trees of bounded rank.
In the following we construct a family of binary trees $(t_n)_{n\in\ns}$
where a smallest TSLP for $t_n$ is exponentially larger than the size of 
a smallest SLP for $t_n$.
Afterwards we show that it is always possible to transform an SLP $\bb A$ for $t$
into a TSLP of size $\mc O(|\bb A| \cdot h \cdot r)$ for $t$,
where $h$ is the height of $t$
and $r$ is the maximal rank of a label occurring in $t$.

\subsubsection{Worst-case comparison of SLPs and TSLPs}

We use the following result from~\cite{DBLP:conf/ifipTCS/BertoniCR08} 
 for the previously mentioned worst-case construction of a family of binary trees:

\begin{theorem}[Thm.~2 from \cite{DBLP:conf/ifipTCS/BertoniCR08}]\label{theorem:bertoni}
 For every $n > 0$, there exist words $u_n, v_n \in \{0,1\}^*$ with $|u_n|= |v_n|$
 such that $u_n$ and $v_n$ have SLPs of size $n^{\mc O(1)}$, but the smallest SLP for 
 the convolution $u_n \otimes v_n$ has size $\Omega(2^{n/2})$.\footnote{Actually, in \cite{DBLP:conf/ifipTCS/BertoniCR08}
 the result is not stated for the convolution $u_n \otimes v_n$ but the literal shuffle of $u_n$ and $v_n$ which
 is $u_n[1] v_n[1] u_n[2] v_n[2] \cdots u_n[m] v_n[m]$. But this makes no difference, since the sizes of the smallest
 SLPs for the convolution and literal shuffle, respectively, of two words differ only by multiplicative constants.} 
\end{theorem}
For two given words $u = i_1 \cdots i_n  \in \{0,1\}^*$ and $v = j_1 \cdots j_n \in \{0,1\}^*$ we define the {\em comb tree}
\[
	t(u,v) = f_{i_1}(f_{i_2}(\dots f_{i_n}(\$,j_n) \dots j_2),j_1) 
\]
over the ranked alphabet $\{ f_0, f_1, 0, 1, \$ \}$ where $f_0,f_1$ have rank $2$ and $0,1,\$$ have rank $0$.
See Figure~\ref{fig:comb_tree} for an illustration.
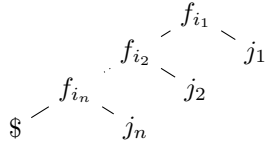
\begin{figure}[t]
\centering

 \begin{tikzpicture}[-,scale=1,level distance=5mm, sibling distance=16mm]
 
 \tikzset{level 1/.style={sibling distance=16mm}}
 \tikzset{level 2/.style={sibling distance=16mm}}
 \tikzset{level 3/.style={sibling distance=16mm}}
 \tikzset{level 4/.style={sibling distance=16mm}}
 
\node (root){$f_{i_1}$}
  child {node (0) {$f_{i_2}$}
    child[dotted] {node (00) {$f_{i_n}$}
       child[solid] {node (000) {$ \$ $}}
       child[solid] {node (001) {$j_n$}}
    }
    child {node (01) {$j_2$}} 
  }
  child {node (1) {$j_1$}}
;
\end{tikzpicture}

\caption{The comb tree $t(u,v)$ for $u = i_1 \cdots i_n$ and $v = j_1 \cdots j_n$ }
 \label{fig:comb_tree}
\end{figure}

\begin{theorem}\label{theorem:SmallestTravGrammar}
 For every $n> 0$ there exists a tree $t_n$ 
 such that  
 the size of a smallest SLP for 
 $t_n$ is polynomial in $n$, but
 the size of a smallest TSLP for $t_n$ is in $\Omega(2^{n/2})$.
\end{theorem}

\begin{proof}
 Let us fix an $n$ and let $u_n$ and $v_n$ 
 be the aforementioned strings from Theorem~\ref{theorem:bertoni}. Let $|u_n|= |v_n|=m$.
 Consider the comb tree $t_n := t(u_n,v_n)$. Note that $t_n = f_{i_1} \cdots f_{i_m}  \$ \, \mathrm{rev}(v_n)$,
 where $u_n = i_1 \cdots i_m$. By Theorem~\ref{theorem:bertoni} there exist SLPs of size 
  $n^{\mc O(1)}$ for $u_n$ and $v_n$, and these SLPs easily yield an SLP of size $n^{\mc O(1)}$ for  $t_n$.
 
 Next, we show that a TSLP $\bb A$ for $t_n$ yields an
  SLP of size $\mc O(|\bb A|)$ for the string $u_n\wedge v_n$.
  Since a smallest SLP for $u_n\wedge v_n$ has size  $\Omega(2^{n/2})$ by Theorem~\ref{theorem:bertoni},
  the same bound must hold for the size of a smallest TSLP for $t_n$. 
 
  Let $\bb A$ be a TSLP for $t_n$.
 By \cite{LoMaSS12} we can transform $\bb A$ into a monadic TSLP $\bb A'$ for $t_n$ of size $\mc O(|\bb A|)$.
 We transform the TSLP $\bb A'$
 into an SLP of the same size for $u_n \otimes v_n$.
 We can assume that every nonterminal except for the start nonterminal $S$ occurs in a right-hand side and 
 every nonterminal occurs in the derivation starting from $S$.
At first we delete all rules of the form $A \to j$ ($j \in \{0,1\}$) and 
 replace the occurrences of $A$ by $j$ in all right-hand sides.
 Now every nonterminal  $A \neq S$ of rank $0$ derives to a subtree of $t_n$ that
 contains the unique $\$$-leaf of $t_n$. Hence, $t_n$ contains a unique subtree $\val(A)$. This implies
 that $A$ occurs exactly once
 in a right hand side. We can therefore without size increase replace this occurrence of $A$
 by the right-hand side of $A$.  
 After this step, $S$ is the only rank-$0$ nonterminal in the TSLP.
With the same argument, we can also eliminate rank-$1$ nonterminals that derive to a tree containing the unique leaf $\$$.
After this step, every rank-$1$ nonterminal $A(x)$ 
derives a tree of the form $g_1(g_{2}(\dots(g_k(x,j_k)\dots),j_2),j_1)$ ($g_i \in \{ f_0,f_1\}$ and 
$j_i \in \{0,1\}$). 

Now, if a right-hand side contains a subtree $f_i(s_1, s_2)$, then $s_2$ must be either $0$ or $1$.
Similarly, for every occurrence of $i \in \{0,1\}$ in a right-hand side, the parent node of that occurrence must be
either labelled with $f_0$ or $f_1$ (note that the parent node exists and cannot be a nonterminal). Therefore we can obtain an SLP for $u_n \otimes v_n$ by replacing every production
$A(x) \to t(x)$ by $A \to \lambda(t(x))$, where $\lambda(t(x))$ is the string obtained inductively by $\lambda(x) = \varepsilon$,
$\lambda(B(s(x)) = B \lambda(s(x))$ for nonterminals $B$, and $\lambda( f_i(s(x),j)) = (i, j) \lambda(s(x))$.
The production for $S$ must be of the form $S \to t(\$)$ for a term $t(x)$ and we replace it
by $S \to \lambda(t(x)) \$$.
\end{proof}

\subsubsection{Conversion of SLPs to TSLPs}
Note that the height of the tree $t_n$ in Theorem~\ref{theorem:SmallestTravGrammar} is linear in the size of $t_n$.
By  the  following result, large height and rank are always responsible for the exponential succinctness gap between SLPs and TSLPs.

\begin{theorem}\label{theorem:TraversalToTree}
 Let $t \in \mc T(\mc F)$ be a tree of height $h$ and maximal rank $r$,
 and let $\bb A$ be an SLP for $t$ with $|\bb A | = m$.
 Then there exists a TSLP $\bb B$ with $\val(\bb B) = t$ such that
 $|\bb B| \in \mc O(m \cdot h \cdot r)$, 
 which can be constructed in time $\mc O(m \cdot h \cdot r)$.
\end{theorem}
 
\begin{proof}
 Without loss of generality we assume that $\bb A$ is in Chomsky normal form.
 For every nonterminal $A$ of $\bb A $ with $c(A) = (a_1,a_2)$
 we introduce $a_1$  nonterminals $A_1, \dots, A_{a_1}$ of rank $0$ 
 (these produce one tree each) and, if $a_2 > 0$, one
 nonterminal $A'$ of rank $a_2$ for the fragment encoded by $A$.
 For every rule of the form $A \to f$ with $f \in {\mc F}_n$
 we add to $\bb B$ the TSLP-rule
 $A_1 \to f$ if $n=0$ or
 $A'(x_1,\dots, x_n) \to f(x_1, \dots, x_n)$  if $n \geq 1$.
 Now consider a rule of the form $A \to BC $
 with $c(B)=(b_1,b_2)$ and $c(C) = (c_1, c_2)$.
 \paragraph{Case 1:}
  If $b_2 = 0$ we add  the following rules to $\bb B$:
  \begin{alignat*}{2}
  	A_i &\to B_i & & \quad \text{for } 1 \le i \le b_1 \\
  	A_{b_1+i} &\to C_i & & \quad \text{for } 1 \le i \le c_1 \\
  	A'(x_1, \dots, x_{c_2}) &\to C'(x_1, \dots, x_{c_2}) & & \quad \text{if } c_2 > 0
  \end{alignat*} 
 \paragraph{Case 2:}
  If $0 < b_2 \le c_1$ we add  the following rules to $\bb B$:
  \begin{alignat*}{2}
  	A_i &\to B_i & & \quad \text{for } 1 \le i \le b_1 \\
  	A_{b_1+1} &\to B'(C_1, \dots, C_{b_2}) & & \\
  	A_{b_1+1+i} &\to C_{b_2+i} & & \quad \text{for } 1 \le i \le c_1-b_2 \\
  	A'(x_1, \dots, x_{c_2}) &\to C'(x_1, \dots, x_{c_2}) & & \quad \text{if } c_2 > 0
  \end{alignat*}
 \paragraph{Case 3:} 
  If $b_2 > c_1$ we add  the following rules to $\bb B$, where $d = b_2-c_1$: 
  \begin{alignat*}{2}
   	A_i &\to B_i & & \quad \text{for } 1 \le i \le b_1 \\
	A'(x_1, \dots, x_d) &\to B'(C_1,\dots, C_{c_1},x_{1}, \dots, x_d) & & \quad \text{if } c_2 = 0 \\
   	A'(x_1, \dots, x_{c_2+d-1}) &\to B'(C_1,\dots, C_{c_1},C'(x_1,\dots,x_{c_2}),x_{c_2+1}, \dots, x_{c_2+d-1}) & & \quad \text{if } c_2 > 0
\end{alignat*}
Chain productions, where the right-hand side consists of a single nonterminal, can be eliminated without size increase.
Then, only one of the above productions remains and its size is  bounded by $c_1 + 2$
(recall that we do not count parameters). Recall that $c_1$ is the number of complete trees produced by $C$.
It therefore suffices to show that the number of complete trees of a factor $s$ of $t$ is bounded by
$h \cdot r$, where $h$ is the height of $t$ and $r$ is the maximal rank of a label in $t$. Assume that $s = t[i:j] = t_1 \cdots t_n s'$, where
$t_i \in \mc T(\mc F)$ and $s'$ is a fragment. Let $k$ be the lowest common ancestor of $i$ and $j$. If $k = i$ (i.e., $i$ is an ancestor of $j$)
then either $s=t_1$ or $s = s'$. Otherwise, the root of every tree $t_l$ ($1 \leq l \leq n$)
is a child of a node on the path from $i$ to $k$. 
The length of the path from $i$ to $k$ is bounded by $h$, hence $n \leq h \cdot r$.
\end{proof}


\subsection{SLPs for traversals versus balanced parenthesis sequences}

Balanced parenthesis sequences are widely used as a succinct representation of ordered unranked unlabeled trees \cite{MunroR01}.
One defines the balanced parenthesis sequence $\mathsf{bp}(t)$ of such a tree $t$ inductively as follows.
If $t$ consists of a single node, then $\mathsf{bp}(t) = ()$. If the root of $t$ has $n$ children in which the subtrees
$t_1, \ldots, t_n$ are rooted (from left to right), then $\mathsf{bp}(t) = (\mathsf{bp}(t_1) \cdots \mathsf{bp}(t_n))$.
Hence, a tree with $n$ nodes is represented by $2n$ bits, which is optimal in the information theoretic sense.
On the other hand, an unlabelled full binary tree $t$ (i.e., a tree where every non-leaf node has exactly two children)
of size $n$ can be represented with $n$ bits by viewing $t$ as a ranked tree over $\mc F = \{a,f\}$, where $f$ has rank two
and $a$ has rank zero.

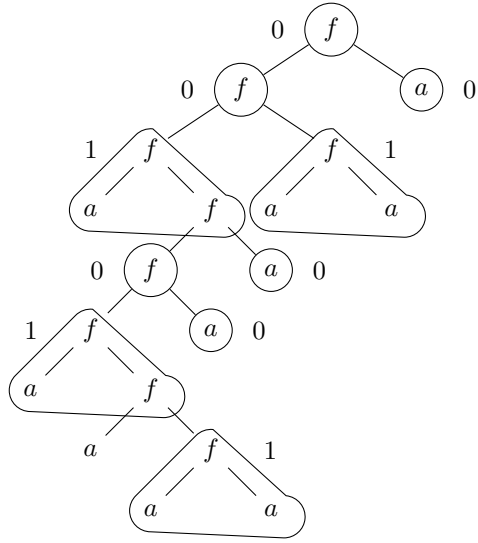
\begin{figure}[t]
\centering
\tikzset{level 1/.style={sibling distance=24mm}}
\tikzset{level 3/.style={sibling distance=16mm}}
\begin{tikzpicture}[scale=1,auto,swap,level distance=8mm]
\node[draw,circle] (f0) {$f$} 
  child {node[draw,circle] (f1) {$f$}
    child {node (f2) {$f$}
      child {node (a1) {$a$}}
      child {node (f3) {$f$}
        child {node[draw,circle] (f4) {$f$}
          child {node (f5) {$f$}
			child {node (a2) {$a$}}
			child {node (f6) {$f$}
			  child {node (a3) {$a$}}
			  child {node (f7) {$f$}
			    child {node (a4) {$a$}}
			    child {node (a5) {$a$}}
			  }
			}
          }
          child {node[draw,circle] (a6) {$a$}}
        }
        child {node[draw,circle] (a7) {$a$}}
      }
    }
    child {node (f8) {$f$}
      child {node (a8) {$a$}}
      child {node (a9) {$a$}}
    }
  }
  child {node[draw,circle](a10) {$a$}}
;
\draw \convexpath{f2,f3,a1}{8pt};
\draw \convexpath{f5,f6,a2}{8pt};
\draw \convexpath{f7,a5,a4}{8pt};
\draw \convexpath{f8,a9,a8}{8pt};

\node [left = 4pt of f0] {$0$};
\node [left = 4pt of f1] {$0$};
\node [left = 10pt of f2] {$1$};
\node [left = 4pt of f4] {$0$};
\node [left = 10pt of f5] {$1$};
\node [right =10pt of f7] {$1$};
\node [right = 4pt of a6] {$0$};
\node [right = 4pt of a7] {$0$};
\node [right =10pt of f8] {$1$};
\node [right = 4pt of a10] {$0$};

\end{tikzpicture}
\caption{Example tree for the proof of Theorem~\ref{theorem:BP}}
 \label{fig:thm5}
\end{figure} 

\begin{theorem}\label{theorem:BP}
 For every $n> 0$ there exists a full binary tree $t_n$ 
 such that  
 the size of a smallest SLP for 
 $t_n$ is polynomial in $n$, but
 the size of a smallest SLP for $\mathsf{bp}(t_n)$ is in $\Omega(2^{n/2})$.
\end{theorem}

\begin{proof}
 Let us fix an $n$ and let $u_n, v_n \in \{0,1\}^*$ 
 be the  strings from Theorem~\ref{theorem:bertoni}. Let $|u_n|= |v_n|=m$.
 We define $t_n$ by 
 $$
 t_n = \varphi_1(\mathrm{rev}(u_n)) \, a \, \varphi_2(v_n) 
 $$
 where $\varphi_1, \varphi_2 : \{0,1\}^* \to \{a,f\}^*$
 are the homomorphisms defined as follows:
 \begin{alignat*}{2}
\varphi_1(0) & =  f &  \qquad \varphi_2(0) & = a \\
\varphi_1(1) & =  faf &  \qquad \varphi_2(1) & = faa
\end{alignat*}
It is easy to see that $t_n$ is indeed a tree (note that the string 
$\varphi_2(v_n)$ is a sequence of $m$ many trees).
From the SLPs for $u_n$ and $v_n$ we obtain an SLP for $t_n$ of size polynomial in $n$.
It remains to show that the smallest SLP for $\mathsf{bp}(t_n)$ has size $\Omega(2^{n/2})$.
To do so, we show that from an SLP for $\mathsf{bp}(t_n)$ we can obtain with a linear size increase
an SLP for the convolution of $u_n$ and $v_n$. In fact, we show the following claim:

\medskip
\noindent
{\em Claim.} The convolution $u_n \otimes v_n$ 
can be obtained from a suffix of $\mathsf{bp}(t_n)$ by a fixed rational transformation
 (i.e., a deterministic finite automaton
that outputs along every transition a finite word over some output alphabet). 

\medskip
\noindent
This claim proves the theorem using the following two facts:
\begin{itemize}
\item An SLP for a suffix of a string $\val(\bb A)$ (for an SLP $\bb A$) can be produced by an SLP
of size $\mc O(|\bb A|)$ by point 4 of Lemma~\ref{lemma:slp_folklore}.
\item  For every fixed rational transformation $\rho$,
an SLP for $\rho(\val(\bb A))$ 
can be produced by an SLP of size $\mc O(|\bb A|)$ \cite[Theorem~1]{DBLP:conf/ifipTCS/BertoniCR08}
(the $\mc O$-constant depends on the rational transformation).
\end{itemize}
To see why the above claim holds, it is the best to look at an example.
Assume that $u_n = 10100$ and $v_n = 10010$. Hence, we have
$$
 t_n = \varphi_1(\mathrm{rev}(u_n)) \, a \, \varphi_2(v_n)\;=\; f \; f \; faf \; f \; faf \; a \; faa \; a \; a \; faa \; a .
$$
This tree is shown in Figure~\ref{fig:thm5}.
We have
$$
\mathsf{bp}(t_n) =   \underset{0}{(} \; \underset{0}{(} \; \underbrace{( () (}_{1}  \; \underset{0}{(}   \; \underbrace{( () (}_{1}  \; () \; 
 \underbrace{(()())))}_{(1,1)}  \; \underbrace{())}_{(0,0)}   \; \underbrace{()))}_{(1,0)}   \; \underbrace{(()()))}_{(0,1)}   \; \underbrace{())}_{(0,0)} .
$$
Indeed, $\mathsf{bp}(t_n)$ starts with an encoding of the string  $\mathrm{rev}(u_n)$ (here $00101$)
via the correspondence $0 \;\widehat{=}\; ($ and $1 \;\widehat{=}\; (()($, followed by $()$  (which encodes the single $a$ between
$\varphi_1(\mathrm{rev}(u_n))$ and  $\varphi_2(v_n)$ in $t_n$), followed by the desired encoding of the convolution
$u_n \otimes v_n$. The latter is encoded by the following correspondence:
\begin{eqnarray*}
(0,0) & \widehat{=} & ()) \\
(1,0) & \widehat{=} & ())) \\
(0,1) & \widehat{=} & (()())) \\
(1,1) & \widehat{=} & (()()))) .
\end{eqnarray*}
So, a $0$ (resp., $1$) in the second component is encoded by $()$ (resp., $(()())$), which corresponds to the tree $a$ (resp., $faa$).
A $0$ (resp., $1$) in the first component is encoded by one (resp., two) closing parenthesis.

Note that the strings $()), ())), (()())), (()())))$ form a prefix code. This allows to replace these strings by the convoluted symbols
$(0,0), (1,0), (0,1)$, and $(1,1)$, respectively, by a deterministic rational transducer. This shows the above claim.
 \end{proof}
 Theorem~\ref{theorem:BP} can be also interpreted as follows: For every $n> 0$ there exists a full binary tree $t_n$ 
 such that  the size of the smallest SLP for the depth-first-unary-degree-sequence (DFUDS -- it is defined in 
 the proof of Theorem~\ref{thm:navi} below) of $t_n$ is polynomial in $n$, but
 the size of the smallest SLP for the balanced parenthesis representation of $t_n$ is in $\Omega(2^{n/2})$.
It remains open, whether there is also a tree family where the opposite situation arises.

\section{Algorithmic problems on SLP-compressed trees}

In this section we study the complexity of several basic algorithmic problems on trees that are represented by SLPs.

\subsection{Efficient tree operations}
 
In \cite{BLRSSW15} it is shown that 
for a given SLP $\bb A$ of size $n$ that produces the balanced parenthesis
representation of an unranked tree $t$ of size $N$, one can produce in time $\mc O(n)$ 
a data structure of size $\mc O(n)$ that supports navigation as well as other 
important tree queries (e.g. lowest common ancestors queries) 
in time $O(\log N)$. Here, the word RAM model is used, where memory cells
can store numbers with $\log N$ bits and arithmetic operations on $\log N$-bit 
numbers can be carried out in constant time.
An analogous result was shown in  \cite{BilleGLW13,Hubschle-Schneider15} for top dags.
Here, we show the same result for SLPs that produce (preorder traversals
of) ranked trees. Recall that we identify the nodes of a tree $t$
with the positions $1, \ldots, |t|$ in the string $t$.

\begin{theorem}
	\label{thm:navi}
	Given an SLP $\bb A$ of size $n$ for a tree $t \in \mc T(\mc F)$ of size $N$,
	one can produce in time $\mc O(n)$ 
        a data structure of size $\mc O(n)$ that allows to do the following computations
        in time  $\mc O(\log N) \leq \mc O(n)$ on a word RAM, 
        where  $i,j,k \in \ns$ with $1 \leq i,j \leq N$ are given in binary notation:
	\begin{enumerate}[(a)]
		\item\label{thm:navi:parent} Compute the parent node of node $i>1$ in $t$.
		\item Compute the  $k^{\mathrm{th}}$ child of  node $i$ in $t$,
		      if it exists.
		\item Compute the number $k$ such that $i>1$ is the $k^{\mathrm{th}}$ child of its parent node.
		\item Compute the size of the subtree rooted at node $i$.
		\item Compute the lowest common ancestor of nodes $i$ and $j$ in $t$.
	\end{enumerate}
\end{theorem}

\begin{proof}
In \cite{BLRSSW15}, it is shown that for an SLP $\bb A$ of size $n$ that produces a well-parenthesized string $w \in \{ (, ) \}^*$
of length $N$, one can produce in time $\mc O(n)$ 
a data structure of size $\mc O(n)$ that allows to do the following computations in time  $O(\log N)$ on a word RAM,
where $1 \leq k,j \leq  N$ are given in binary notation and $b \in \{ (, ) \}$:
\begin{itemize}
\item Compute the number of positions $1 \leq i \leq k$ such that $w[i] = b$ ($\mathsf{rank}_b(k)$).
\item Compute the position of the $k^{\mathrm{th}}$ occurrence of $b$ in $w$ if it exists ($\mathsf{select}_b(k)$).
\item Compute the position of the matching closing (resp., opening) parenthesis for an opening (resp., closing)
parenthesis at position $k$ ($\mathsf{findclose}(k)$ and $\mathsf{findopen}(k)$).
\item Compute the left-most position $i \in [k,j]$ having the smallest excess value in the interval $[k,j]$, where the 
excess value at a position $i$ is $\mathsf{rank}_{(}(i)-\mathsf{rank}_{)}(i)$ ($\mathsf{rmqi}(k, j)$). 
\end{itemize}
Let us now take an SLP $\bb A$ of size $n$ for a tree $t \in \mc T(\mc F)$ of size $N$ and let $s$ be the corresponding unlabelled
tree. In \cite{BenoitDMRRR05}, the DFUDS-representation (DFUDS for depth-first-unary-degree-sequence) 
of $s$ is defined as follows: Walk over the tree in preorder and write down for every node with $d$ children
the string $(^d )$ ($d$ opening parenthesis followed by a closing parenthesis). Finally put an additional opening
parenthesis at the beginning of the resulting string, which yields a well-parenthesized string.
For instance, for the tree $g(f(a,a),a,h(a))$ we obtain
the DFUDS-representation $( \; ((() \, (() \; )\; ) \; ) \; () \; )$. Clearly, from the SLP $\bb A$ we can produce 
an SLP $\bb B$ for the DFUDS-representation of the tree $s$: Simply replace in right-hand sides every occurrence of a
symbol $f$ of rank $d$ by $(^d )$, and add an opening parenthesis in front of 
the right-hand side of the start nonterminal.

The starting position of the encoding of a node $i \in \{1, \ldots, N\}$
in the  DFUDS-representation can be found as $\mathsf{select}_{)}(i-1)+1$ for $i >1$, and for $i=1$ it is $2$.
Vice, versa if $k$ is the starting position of the encoding of a node in the DFUDS-representation, then
the preorder number of that node is $\mathsf{rank}_{)}(k-1)+1$.

In \cite{BenoitDMRRR05,JanssonSS12}, it is shown that the tree navigation operations from the theorem can be implemented on the DFUDS-representation
using a constant number of $\mathsf{rank}$, $\mathsf{select}$, $\mathsf{findclose}(k)$, $\mathsf{findopen}(k)$ and $\mathsf{rmqi}$-operations. 
Together with the above mentioned results from  \cite{BLRSSW15} this shows the theorem.
\end{proof}
The data structure of \cite{BLRSSW15} allows to compute the depth and height of a given tree node in time $\mc O(\log N)$ 
as well. It is not clear to us, whether this result can be extended to our setting as well. In
 \cite{JanssonSS12}  it is shown that the depth of a given node can be computed in constant time on the 
 DFUDS-representation. But this uses an extra data structure, and it is not clear whether this extra
 data structure can be adapted so that it works for an SLP-compressed DFUDS-representation.
 On the other hand, in Section~\ref{sec-tree-eval}, we show that the height and depth of a given node of an SLP-compressed
 tree can be computed in polynomial time at least.

\subsection{Pattern matching}
\label{pattern-matching}

In contrast to navigation problems, simple pattern matching problems become quite difficult for SLP-compressed trees.
The {\em pattern matching problem for SLP-compressed trees} can be formalized as follows: 
Given a tree $s \in \mc T(\mc F \cup \mc X)$, called the {\em pattern}, where every parameter $x \in \mc X$ occurs at most once,
and an SLP $\bb A$ producing a tree $t \in \mcT$, is there a substitution $\sigma: \mc X \to \mc T(\mc F)$ such that $\sigma(s)$ is a subtree of $t$? 
Here, $\sigma(s) \in \mc T(\mc F)$ denotes the tree obtained from $s$ by substituting each variable $x \in \mc X$ by the tree $\sigma(x)$.
Note that the pattern is given in uncompressed form. If the tree $t$ is given by a TSLP, the corresponding problem can be solved
in polynomial time \cite{DBLP:journals/corr/abs-1302-6336} (even if the pattern tree $s$ is given by a TSLP as well).

\begin{theorem}
	\label{thm:pattern}
	The pattern matching problem for SLP-compressed trees is \NP-complete. Moreover, $\NP$-hard\-ness holds for a fixed
	pattern of the form  $f(x,a)$
\end{theorem}

\begin{proof}
	The problem is contained in \NP because one can guess a node $i \in \{ 1, \dots, |t| \}$ 
	and verify whether the subtree of $t$ rooted in $i$ matches the pattern $s$. 
	The verification is possible in polynomial time by comparing all relevant symbols using Theorem \ref{thm:navi}.

	By \cite[Theorem 3.13]{LoCWP} it is \NP-complete to decide for given SLPs 
	$\bb A,	\bb B$ over $\{0,1\}$ with $|\val(\bb A)| = |\val(\bb B)|$ whether there exists a position 
	$i$ such that $\val(\bb A)[i] = \val(\bb B)[i] = 1$. 
	This question can be reduced to the pattern matching problem with a fixed pattern. 
	One can compute in polynomial time from $\bb A$ and $\bb B$ an SLP $\bb T$ 
	for the comb tree $t(\val(\bb A),\val(\bb B))$. 
	There exists a position $i$ such that $\val(\bb A)[i] = \val(\bb B)[i] = 1$ 
	if and only if the pattern $f_1(x,1)$ occurs in $t(\val(\bb A),\val(\bb B))$.
\end{proof}

\subsection{Tree evaluation problems} \label{sec-tree-eval}

The algorithmic difficulty of SLP-compressed trees already becomes clear when computing the height.
For TSLPs it is easy to see that the height of the produced tree can be computed in linear time:
Compute bottom-up for each nonterminal the height of the produced tree and the depths of the parameter nodes.
However, this direct approach fails for SLPs since each nonterminal encodes a possibly exponential number of trees.
The crucial observation to solve this problem is that one can store and compute the required information
for each nonterminal in a compressed form.

In the following we present a general framework to define and solve evaluation problems on SLP-compressed trees.
We assign to each alphabet symbol of rank $n$ an $n$-ary operator which defines the value of a tree by evaluating it bottom-up.
This approach includes natural tree problems like computing the height of a tree, evaluating a Boolean expression or
determining whether a fixed tree automaton accepts a given tree.
We only consider operators on $\zs$ but other domains with an appropriate encoding of the elements are also possible.
To be able to consider arbitrary arithmetic expressions properly, it is necessary to allow 
the set of constants of a ranked alphabet $\mc F$ to be infinite, i.e.\ $\mc F_0 \subseteq \zs$.


\begin{definition} \label{def-interpret} 
        Let $\mc D \subseteq \zs$ be a (possibly infinite) domain of integers and let 
	$\mc F$ be a ranked alphabet with $\mc F_0 = \mc D$. An {\em interpretation 
	$\mc I$ of $\mc F$ over $\mc D$} assigns to each function symbol $f \in \mc F_n$ an $n$-ary function $f^{\mc I}: \mc D^n \to \mc D$ with
	the restriction that $a^{\mc I} = a$ for all $a \in \mc D$.
	We lift the definition of $\mc I$ to $\mc T(\mc F)$ inductively by
	\[
		(f \, t_1 \cdots t_n)^{\mc I} = f^{\mc I}(t_1^{\mc I}, \ldots, t_n^{\mc I}),
	\]
	where $f \in \mc F_n$ and $t_1, \dots, t_n \in \mc T(\mc F)$.  
\end{definition}

\begin{definition} The {\em $\mc I$-evaluation problem for SLP-compressed trees} is the following problem:
	Given an SLP $\bb A$ over $\mc F$ with $\val(\bb A) \in \mcT$, compute $\val(\bb A)^{\mc I}$.
\end{definition}

\subsubsection{Reduction to caterpillar trees}

In this section, we reduce the $\mc I$-evaluation problem for SLP-compressed trees 
to the corresponding problem for SLP-compressed caterpillar trees.
 A tree $t \in \mc T(\mc F)$ is called a {\em caterpillar tree} if every node has at most one child which is not a leaf.
Let $s \in \mc F^*$ be an arbitrary string. Then $s^{\mc I} \in \mc F^*$ denotes the unique string 
obtained from $s$ by replacing every maximal substring $t \in  \mcT$ of $s$ by
its value $t^{\mc I}$. By Lemma~\ref{lemma:substring_traversal_string} we can factorize $s$ 
uniquely as
$s = t_1 \cdots t_n u$ where $t_1, \ldots, t_n \in \mcT$ and $u$ is a fragment.
Hence $s^{\mc I} = m_1 \cdots m_n u^{\mc I}$ with $m_1, \ldots, m_n \in \mc D$. 
Since $u$ is a fragment, the string $u^{\mc I}$ is the fragment of a caterpillar tree (briefly,
caterpillar fragment in the following).

\begin{example}
Let $\mc F = \{0,1,2,+, \times\}$ with the standard interpretation on integers 
($+$ and $\times$ are considered as binary operators). Consider $s = 0 \, 2 + 2 + + \times 2 + 2 \, 1 + \times$.
Since $+21$ evaluates to $3$, and $\times 2 3$ evaluates to $6$, we have
 $s^{\mc I} =  0 \, 2 + 2 + +\, 6 + \times$. 
\end{example}
Our reduction to caterpillar trees only works for interpretations that satisfy a certain growth condition.
We say that an interpretation $\mc I$ is {\em polynomially bounded}, if there exist 
constants $\alpha,\beta \geq 0$ such that for every tree $t \in \mcT$ (we denote the absolute value of an integer
by $z$ by $\abs{z}$ instead of $|z|$ in order to get not confused with the size $|t|$ of a tree), 
\[
\abs{t^{\mc I}} \le \left( \beta \cdot |t| + \sum_{i \in L} \abs{t[i]} \right)^\alpha
\] 
where $L\subseteq \{1,\ldots,|t|\}$ is the set of leaves of $t$.
The purpose of this definition is to ensure that for every SLP $\bb A$ with $\val(\bb A) \in \mcT$, 
both the length of the binary encoding of $\val(\bb A)^{\mc I}$
and the integer constants that appear in $\bb A$
are polynomially bounded in $|\bb A|$.
\begin{theorem}
	\label{thm:main-reduction}
	Let $\mc I$ be a polynomially bounded interpretation.
	Then the $\mc I$-evaluation problem for SLP-compressed trees is polynomial time 
	Turing-reducible to the $\mc I$-evaluation problem for SLP-compressed caterpillar trees.
\end{theorem}

\begin{proof}
	In the proof we use an extension of SLPs by the cut-operator, called {\em composition systems}. A {\em composition system} 
	$\bb A = (N,\Sigma,P,S)$ is an SLP where $P$ may also contain rules of the form $A \to B[i:j]$ where $A,B \in N$ and $i,j \ge 0$. 
	Here we let $\val(A) = \val(B)[i:j]$.
	It is known \cite{Hag00} (see also \cite{LoCWP}) that a given composition system can be transformed in polynomial time into an SLP with the same value.
	One can also allow mixed rules $A \to X_1 \cdots X_n$ where each $X_i$ is either a terminal, 
	a nonterminal or an expression of the form $B[i : j]$, which clearly
	can be eliminated in polynomial time.

	Let $\bb A = (N,\mc F,P,S)$ be the input SLP in Chomsky normal form.
	We use the notation $c(A) = c(\val(\bb A))$ as in the proof of Theorem~\ref{thm-rec-trav}.
	We will compute a composition system where for each nonterminal 
	$A \in N$ there are nonterminals $A_1$ and $A_2$ in the composition system such that
	the following holds: Assume that $\val(A) = t_1 \cdots t_n \, s$, where $t_1, \ldots, t_n \in \mcT$ and 
	$s$ is a fragment. Hence, $c(A) = (n, \gaps{s})$.
 	Then we will have
	\begin{itemize}
		\item $\val(A_1) = t_1^{\mc I} \cdots t_n^{\mc I} \in \mc D^*$, and
		\item $\val(A_2) = s^{\mc I}$.
	\end{itemize}
	In particular, $\val(A_1) \val(A_2) = \val(A)^{\mc I}$ and 
	$\val(\bb A)^{\mc I}$ is given by the single number in $\val(S_1)$.
	
	The computation is straightforward for rules of the form $A \to f$ with $A \in N$ and $f \in \mc F$:
	If $\rk(f) = 0$, then $\val(A_1) = f$ and $\val(A_2) = \varepsilon$.
	If $\rk(f) > 0$, then $\val(A_1) = \varepsilon$ and $\val(A_2) = f$.
	
	For a nonterminal $A \in N$ with the rule $A \to BC$ we make a case distinction depending on $c(B) = (b_1,b_2)$ and $c(C) = (c_1,c_2)$.

	\medskip
	\noindent
	{\em Case}
	$b_2 \le c_1$: Then concatenating $\val(B)$ and $\val(C)$ yields a new tree $t_{\mathrm{new}}$ (or $\varepsilon$ if $b_2 = 0$) in $\val(A)$.
	Note that $t_\mathrm{new}^{\mc I}$ is the value of the tree $\val(B_2)\, \val(C_1)[1 : b_2]$.
	Hence we can compute $t_\mathrm{new}^{\mc I}$ in polynomial time by computing an SLP that produces $\val(B_2) \, \val(C_1)[1 : b_2]$
	and querying the oracle for caterpillar trees. We add the following rules to the composition system:
	\begin{align*}
	A_1 &\to B_1 \, t_{\mathrm{new}}^{\mc I} \, C_1[b_2 + 1 : c_1] \\
	A_2 &\to C_2
	\end{align*}
        
	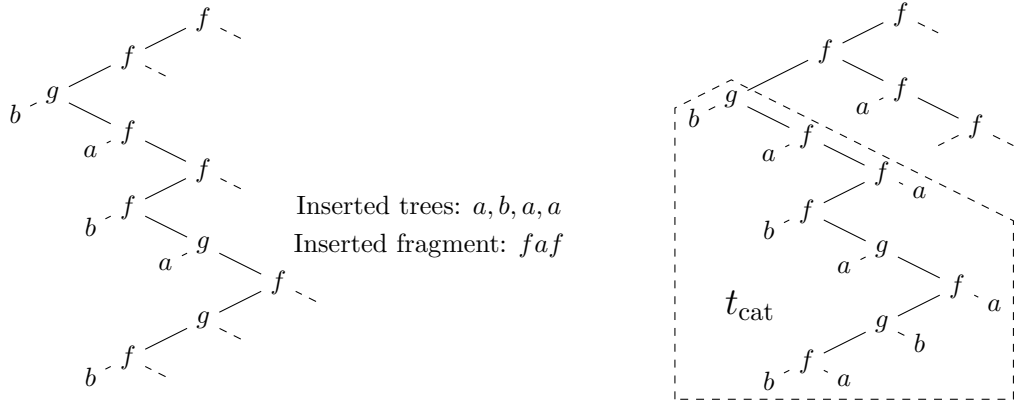
\begin{figure}[t] 
	\centering
	
	\begin{minipage}[hbt]{0.4\textwidth} 	
	\begin{tikzpicture}[scale=1.0]
	 \centering
	
	 \node (n1)  at (1.5,0.5){$f$};
	 \node (n2)  at (2.5,1)  {$g$};
	 \node (n3)  at (3.5,1.5){$f$};
	 \node (n4)  at (2.5,2)  {$g$};
	 \node (n5)  at (1.5,2.5){$f$};
	 \node (n6)  at (2.5,3)  {$f$};
	 \node (n7)  at (1.5,3.5){$f$};
	 \node (n8)  at (0.5,4)  {$g$};
	 \node (n9)  at (1.5,4.5){$f$};
	 \node (n10) at (2.5,5)  {$f$};
	 
	 \draw (n1) -- (n2) -- (n3) -- (n4) -- (n5) 
	    -- (n6) -- (n7) -- (n8) -- (n9) -- (n10);
	
	 \node (a1)  at (1,0.25){$b$};
	 \node (a4)  at (2,1.75){$a$};
	 \node (a5)  at (1,2.25){$b$}; 
	 \node (a7)  at (1,3.25){$a$};
	 \node (a8)  at (0,3.75){$b$};
	 
	 \foreach \a/\n in{
	  {(a1)/(n1)},{(a4)/(n4)},{(a5)/(n5)},{(a7)/(n7)},{(a8)/(n8)}}
	 {
	  \draw \a -- \n;
	 };
	 
	  \foreach \x/ \y/ \node in{
	 {2 / 0.25/ (n1)},
	 {3 / 0.75/ (n2)}, {4 / 1.25/ (n3)},
	 {3 / 2.75/ (n6)}, {2 / 4.25/ (n9)},
	 {3 / 4.75/ (n10)}}
	 {
	  \draw[dashed] (\x,\y) -- \node;
	  };
	 
	 \node (text) at (5.5,2.5) {Inserted trees: $a,b,a,a$};
	 
	 \node (text2) at (5.5,2) {Inserted fragment: $faf$};
	 
	 \end{tikzpicture}
	\end{minipage}
	\hfill
	\begin{minipage}[hbt]{0.4\textwidth} 
	\begin{tikzpicture}[scale=1.0]
	 \centering
	
	 \node (n1)  at (1.5,0.5){$f$};
	 \node (n2)  at (2.5,1)  {$g$};
	 \node (n3)  at (3.5,1.5){$f$};
	 \node (n4)  at (2.5,2)  {$g$};
	 \node (n5)  at (1.5,2.5){$f$};
	 \node (n6)  at (2.5,3)  {$f$};
	 \node (n7)  at (1.5,3.5){$f$};
	 \node (n8)  at (0.5,4)  {$g$};
	 \node (n9)  at (1.75,4.625){$f$};
	 \node (n10) at (2.75,5.125)  {$f$};
	 
	 \draw (n1) -- (n2) -- (n3) -- (n4) -- (n5) 
	    -- (n6) -- (n7) -- (n8) -- (n9) -- (n10);
	
	 \node (a1)  at (1,0.25){$b$};
	 \node (a4)  at (2,1.75){$a$};
	 \node (a5)  at (1,2.25){$b$}; 
	 \node (a7)  at (1,3.25){$a$};
	 \node (a8)  at (0,3.75){$b$};
	 
	 \node (a1r) at (2,0.25){$a$};
	 \node (a2)  at (3,0.75){$b$};
	 \node (a3)  at (4,1.25){$a$};
	 \node (a6)  at (3,2.75){$a$};

	 \foreach \a/\n in{
	  {(a1)/(n1)},{(a4)/(n4)},{(a5)/(n5)},{(a7)/(n7)},{(a8)/(n8)},
	  {(a1r)/(n1)},{(a2)/(n2)},{(a3)/(n3)}, {(a6)/(n6)}}
	 {
	  \draw \a -- \n;
	 };
	 
	 \node (f1) at (2.75,4.125){$f$};
	 \node (f2) at (2.25,3.875) {$a$};
	 \node (f3) at (3.75,3.625){$f$};
	 
	 \draw (n9) -- (f1)-- (f2);
	 \draw (f1) -- (f3);
	 \draw[dashed] (f3) -- (3.25,3.375);
	 \draw[dashed] (f3) -- (4.25,3.375);

	 \draw[dashed] (n10) -- (3.25,4.875);
	  
	 \draw[dashed] (-0.25,0) -- (-0.25,3.875) -- (0.5,4.25) -- (4.25,2.375)
	     -- (4.25,0)  -- (-0.25,0);
	  
	 \node (text) at (0.75,1.25) {\Large $t_{\mathrm{cat}}$};
	  
	 \end{tikzpicture}
	\end{minipage}
	 \caption{An example for case 2 in the proof of Theorem~\ref{thm:main-reduction}.
	 In the left fragment we insert the trees $a,b,a,a$ and the fragment $faf$.
	 The latter yield, together with a part of the fragment, a new tree $t_{\mathrm{cat}}$.
	 }
	 \label{fig:caterpillar}
	\end{figure}

	\noindent
	{\em Case}
	$b_2 > c_1$: Then all trees and the fragment produced by $C$ are inserted into the gaps of the fragment encoded by $B$.
	If $c_1 = 0$ (i.e., $\val(C_1) = \varepsilon$), then we add the productions $A_1 \to B_1$ and $A_2 \to B_2 C_2$. Now assume 
	that $c_1 > 0$.
	Consider the fragment
	\[
	 s = \val(B_2) \, \val(C_1) \, \val(C_2).
	\]
	Intuitively, this fragment $s$ is obtained
	by taking the caterpillar fragment  $\val(B_2)$, where the first $c_1$ many gaps are replaced by the constants
	from the sequence $\val(C_1)$ and the $(c_1+1)^{\text{st}}$ gap is replaced by the caterpillar fragment $\val(C_2)$, see Figure~\ref{fig:caterpillar}.
	If $s$ is not already a caterpillar fragment, then we have to replace the (unique) largest factor of 
	$s$ which belongs to $\mcT$  by its value under $\mc I$ to get $s^{\mc I}$.	
	To do so we proceed as follows:
	Consider the tree $t' = \val(B_2) \, \val(C_1) \, \diamond^{b_2-c_1}$,
	where $\diamond$ is an arbitrary symbol of rank 0, and let $r = |\val(B_2)|+c_1+1$ (the position
	of the first $\diamond$ in $t'$).
	Let $q$ be the parent node of $r$, which can be computed in polynomial time by Theorem~\ref{thm:navi}.
	Using Lemma~\ref{lemma:firstpos} we compute the position $p$ of the first occurrence of a symbol in $t'[q+1:]$ with rank $> 0$.
	If no such symbol exists, then $s$ is already a caterpillar fragment and
	we add the rules $A_1 \to B_1$ and $A_2 \to B_2 C_1 C_2$ to the composition system.
	Otherwise $p$ is the first symbol of the largest factor from $\mc T(\mc F)$ described above. Using Theorem~\ref{thm:navi}(d),
	we can compute in polynomial time the last position $p'$ of the subtree of $t'$ that is rooted in $p$. 
	Note that the position $p$ must belong to $\val(B_2)$ and that $p'$ must belong to $\val(C_1)$ (since 
	$c_1 > 0$). The string $t_{\mathrm{cat}} = (\val(B_2)\, \val(C_1))[p:p']$ is a caterpillar tree  for which we can compute an SLP in polynomial
	time by the above remark on composition systems. Hence, using the oracle we can compute the value  $t_{\mathrm{cat}}^{\mc I}$.	
	We then add the rules 
	\begin{eqnarray*}
         A_1 & \to & B_1, \\
         A'    & \to & B_2 C_1, \text{ and } \\
         A_2    & \to & A'[:p-1]  \, t_{\mathrm{cat}}^{\mc I} \, A'[p'+1:] \, C_2
        \end{eqnarray*}
	to the composition system.
	This completes the proof.
\end{proof}

\subsubsection{Polynomial time solvable evaluation problems}

Next, we present several applications of Theorem~\ref{thm:main-reduction}.
We start with the height of a tree.

\begin{theorem} \label{thm-height}
	The height of a tree $t \in \mc T(\mc F)$ given by an SLP $\bb A $ is computable in polynomial time.
\end{theorem}

\begin{proof}
        We can assume that $t$ is not a single constant. We replace every symbol in $\mc F_0$
        by the integer $0$. Then,
	the height of $t$ is given by its value under the interpretation $\mc I$
	with $f^{\mc I}(a_1, \dots, a_n) = 1 + \max\{a_1, \dots, a_n\}$
	for symbols $f \in \mc F_n$ with $n >0$. 
	Clearly, $\mc I$ is polynomially bounded. 
	By Theorem~\ref{thm:main-reduction} it is enough to show how to evaluate a caterpillar tree $t$ given by an SLP $\bb A$ 
	in polynomial time under the interpretation $\mc I$. But note that in this caterpillar tree, arbitrary natural numbers may occur at leaf positions.
	
	Let $\mc D_t = \{ d \in \ns \mid d\text{ labels a leaf of } t \}$. The size of this set is bounded
	by  $|\bb A|$.
	For $d\in \mc D_t$ let $v_d$ 
	be the largest (i.e., deepest) node 
	such that $d$ is the label of a child of node $v_d$ (in particular, $v_d$ is not a leaf).
	Let us first argue that $v_d$ can be computed in polynomial time.
	
	Let $k$ be the maximal position in $t$
	where a symbol of rank larger than zero occurs. 
	The number $k$ is computable in polynomial time by 
	Lemma~\ref{lemma:slp_folklore} (point~\ref{lemma:number-of-occ} and~\ref{lemma:firstpos}). 
	Again using Lemma~\ref{lemma:slp_folklore}
	we compute  the position of $d$'s last (resp., first) occurrence in $t[:k]$ (resp., $t[k+1:]$). 
	Then using Theorem~\ref{thm:navi} we compute the parent nodes of those two nodes in 
	$t$ and take the maximum (i.e., the deeper one) of both.
	This node is $v_d$.
	
	 Assume that $\mc D_t = \{ d_1, \ldots, d_m\}$, where w.l.o.g. $v_{d_1} < v_{d_2} < \cdots < v_{d_m}$
	 (if $v_{d_i} = v_{d_j}$ for $d_i < d_j$, then we simply ignore $d_i$ in the following consideration). Note 
	 that $v_{d_m}$ is the maximal position in $t$
	 where a symbol of rank larger than zero occurs (called $k$ above).
	 Let $t_i$ be the subtree rooted at $v_{d_i}$. Then $t_m^{\mc I} = d_m+1$. 
	 We now claim that from the value $t_{i+1}^{\mc I}$ we can compute in polynomial time the 
         value $t_i^{\mc I}$. The crucial point is that we can ignore all constants that appear in the interval
         $[v_{d_{i}}+1, v_{d_{i+1}}-1]$ except for $d_i$. More precisely,
         assume that $a = t_{i+1}^{\mc I}$ and let $b$ be the number of occurrences
         of symbols of rank at least one in the interval $[v_{d_{i}}+1, v_{d_{i+1}}-1]$. Also this number can be computed
         in polynomial time by Lemma~\ref{lemma:slp_folklore}. Then the value of $t_i^{\mc I}$ is $\max\{ a+b+1, d_i+1\}$.
         Finally, using the same argument, we can compute $t^{\mc I}$ from $t_1^{\mc I}$.
\end{proof}

\begin{corollary}
Given an SLP $\bb A$  for a tree $t$ and a node $1 \leq i \leq |t|$ one can compute the depth of $i$ in $t$
in polynomial time.
\end{corollary}

\begin{proof}
We can write $t$ as $t = uvw$, where $|u| = i-1$ and $v$ is the subtree of $t$ rooted at node $i$. We can compute
in polynomial time $|v|$ by Theorem~\ref{thm:navi}. This allows to compute in polynomial an SLP for the tree $u h^{|t|} a w$. Here, $h$ has rank one
and $a$ has rank zero. Then the depth of $i$ in $t$ is $\mathrm{height}(u h^{|t|} a w) - |t|$.
\end{proof}
An interesting parameter of a tree $t$ is its  \emph{Horton-Strahler number} or \emph{Strahler number}, 
see \cite{DBLP:conf/lata/EsparzaLS14} for a recent survey.
It can be defined as the value $t^{\mc I}$ under the 
interpretation $\mc I$ over $\ns$ which interprets constant symbols
$a \in \mc F_0$ by $a^{\mc I} = 0$ and each symbol $f \in \mc F_n$ with $n>0$ as follows:
Let $a_1, \dots, a_n \in \ns$ and $a = \max\{a_1, \dots, a_n\}$. We set
$f^{\mc I}(a_1, \dots, a_n) = a$ if exactly one of $a_1, \dots, a_n$ is equal to $a$,
and otherwise $f^{\mc I}(a_1, \dots, a_n) = a + 1$.
The Strahler number was first defined in hydrology, but also has 
many applications in computer science \cite{DBLP:conf/lata/EsparzaLS14} , e.g.\ to calculate 
the minimum number of registers required to evaluate an arithmetic expression~\cite{DBLP:journals/tcs/FlajoletRV79}.

\begin{theorem}
 Given an SLP $\bb A$ for a tree $t$, one can compute the Strahler number of $t$ in polynomial time.
\end{theorem}

\begin{proof}
  Note that the interpretation $\mc I$ above is very similar to the one from the proof of Theorem~\ref{thm-height}. 
  The only difference is that the uniqueness of the maximum among the children of a node also affects the evaluation.
  Therefore the proof of Theorem~\ref{thm-height} must be slightly modified by considering 
  for each $d\in \mc \ns$ occurring in $t$ the two deepest leaves in $t$ labelled with $d$ (or the unique
  leaf labelled by $d$ if $d$ occurs exactly once).
  Let $i$ and $j$ be the parents of those two leaves $(i\ge j)$ and let $t_i$ 
  (resp., $t_j$) be the subtree of $t$ rooted at $i$ (resp., $j$). The nodes $i$ and $j$ can be computed
  in polynomial time as in the proof of Theorem~\ref{thm-height}.
  We have $t_i^{\mc I} \geq d$, and therefore $t_j^{\mc I}=d+1$. This implies that
  any further occurrence of $d$ that is higher up in the tree has no influence on the evaluation process.
  The rest of the argument is similar to the proof of Theorem~\ref{thm-height}.
\end{proof}
If the interpretation $\mc I$ is clear from the context, 
we also speak of the problem of {\em evaluating SLP-compressed $\mc F$-trees}.
In the following theorem the interpretation is given by the Boolean operations $\wedge$ and $\vee$ over $\{0,1\}$.

\begin{theorem}\label{theorem:eval_bool}
Evaluating SLP-compressed $\{\wedge,\vee,0,1\}$-trees can be done in polynomial time.
\end{theorem}

\begin{proof}
Let $\bb A$ be an SLP over $\{\wedge, \vee, 0, 1\}$ such that $\val(\bb A)$ is a caterpillar tree. 
Define a {\em left caterpillar tree} to be a tree of the form
$uv$, where $u \in \{\wedge, \vee\}^*$, $v \in \{0,1\}^*$ and $|v| = |u|+1$. 
That means that the main branch of the caterpillar tree grows to the left.
The evaluation of $\val(\bb A)$ is done in two steps. 
In a first step, we compute in polynomial time from $\bb A$ a new SLP $\bb B$ such that
$\bb B$ is a left caterpillar tree and $\val(\bb A)^{\mc I} = \val(\bb B)^{\mc I}$. 
In a second step, we show how to evaluate a left caterpillar tree.
We can assume that $\val(\bb A)$ is neither $0$ or $1$.

\medskip
\noindent
{\em Step 1.} (See Figure~\ref{fig:bool} for an illustration of step 1.) Since $\val(\bb A)$ is a caterpillar tree, we have $\val(\bb A)=uv$ with 
$u \in \{ \wedge, \vee, \wedge 0, \wedge 1, \vee 0, \vee 1\}^* \cdot \{\wedge, \vee\}$, $v \in \{0,1\}^*$
and $|v|$ is $1$ plus the number of occurrences of the symbols $\wedge, \vee$ in $u$ that are not followed by
$0$ or $1$ in $u$.
We can compute bottom-up the length of the maximal suffix of $\val(\bb A)$ from $\{0,1\}^*$ in polynomial time. Hence, by 
Lemma~\ref{lemma:slp_folklore} we can compute in polynomial time SLPs $\bb A_1$ and $\bb A_2$ such that
$\val(\bb A_1) = u$ and $\val(\bb A_2) = v$.

We will show how to eliminate all occurrences of the patterns $\wedge 0, \wedge 1, \vee 0, \vee 1$.
For this, it is technically easier to replace every occurrence of $\circ a$ by a new symbol $\circ_a$,
where $\circ \in \{\wedge,\vee\}$ and $a \in \{0,1\}$. Let $\varphi : \{ \wedge, \vee, \wedge 0, \wedge 1, \vee 0, \vee 1\}^*
\to  \{ \wedge, \vee, \wedge_0, \wedge_1, \vee_0, \vee_1\}^*$ be the mapping that replaces every
occurrence of $\circ a$ by the new symbol $\circ_a$ ($\circ \in \{\wedge,\vee\}$, $a \in \{0,1\}$). 
This mapping is a rational transformation. 
Hence, using \cite[Theorem~1]{DBLP:conf/ifipTCS/BertoniCR08}, 
we can compute in polynomial time an SLP $\bb B_1$ for $\varphi(\val(\bb A_1))$. 
We now compute, using Lemma~\ref{lemma:slp_folklore},
the position $i$ in $\val(\bb B_1)$ of the first occurrence of a symbol from $\{ \wedge_0, \vee_1\}$.
Next, we compute an SLP $\bb C_1$ for the prefix $\val(\bb B_1)[:i-1]$, i.e., we cut off the suffix starting
in position $i$. Moreover, we compute the number $j$ of occurrences of symbols from $\{\wedge, \vee\}$ in
the suffix $\val(\bb B_1)[i:]$ and compute an SLP $\bb B_2$ for the string $0 \, \val(\bb A_2)[j+2:]$ in case
$\val(\bb B_1)[i] = \wedge_0$ and $1 \, \val(\bb A_2)[j+2:]$ in case
$\val(\bb B_1)[i] = \vee_1$. Then $\val(\bb A)$ evaluates to the same truth value as $\varphi^{-1}(\val(\bb C_1)) \, \val(\bb B_2)$.
The reason for this is that $\varphi^{-1}(\val(\bb B_1)[i:]) \, \val(\bb A_2)[:j+1]$ is a tree which 
evaluates to $0$ (resp., $1$)
if $\val(\bb B_1)[i] = \wedge_0$ (resp., $\val(\bb B_1)[i] = \vee_1$), because $0 \wedge x = 0$ (resp., $1 \vee x = 1$).

Note that $\varphi^{-1}(\val(\bb C_1)) \, \val(\bb B_2)$ is  a caterpillar tree, where 
$\val(\bb C_1) \in \{ \wedge, \vee, \wedge_1, \vee_0\}^*$ and $\val(\bb B_2) \in \{0,1\}^*$.
Since $1 \wedge x = x$ (resp., $0 \vee x = x$), we can delete in the string $\val(\bb C_1)$ all occurrences
of the symbols $\wedge_1$ and $\vee_0$ without changing the final truth value. Let $\bb D_1$ 
be an SLP for the resulting string, which is easy to compute from $\bb C_1$. 
Then $\val(\bb D_1) \, \val(\bb B_2)$ is indeed a left caterpillar tree.

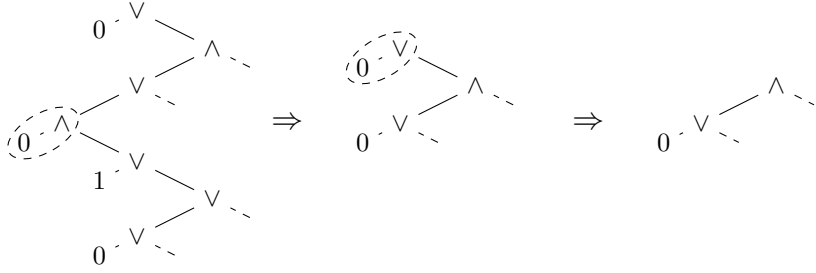
\begin{figure}[t] 
\centering

\begin{tikzpicture}[scale=1.0]

 \node (n1)  at (1.5,0.5){$\vee$};
 \node (n2)  at (2.5,1)  {$\vee$};
 \node (n3)  at (1.5,1.5){$\vee$};
 \node (n4)  at (0.5,2)  {$\wedge$};
 \node (n5)  at (1.5,2.5){$\vee$};
 \node (n6)  at (2.5,3)  {$\wedge$};
 \node (n7)  at (1.5,3.5){$\vee$};

 \draw (n1) -- (n2) -- (n3) -- (n4) -- (n5) 
    -- (n6) -- (n7);
    
 \node (a1)  at (1,0.25){$0$};
 \node (a3)  at (1,1.25){$1$};
 \node (a4)  at (0,1.75){$0$}; 
 \node (a7)  at (1,3.25){$0$};

  \foreach \a/\n in{
  {(a1)/(n1)},{(a3)/(n3)},{(a4)/(n4)},{(a7)/(n7)}}
 {
  \draw \a -- \n;
 };
 
  \draw[rotate around={-60:(0.25,1.875)}, dashed] (0.25,1.875) ellipse (0.26 and 0.52);
 
   \foreach \x/ \y/ \node in{
 {2 / 0.25/ (n1)},
 {3 / 0.75/ (n2)},
 {2 / 2.25/ (n5)},
 {3 / 2.75/ (n6)}}
 {
  \draw[dashed] (\x,\y) -- \node;
  };
  
  \node (to) at (3.5,2) {\large $\Rightarrow$};
  
 \node (a51)  at (4.5,1.75)  {$0$};
 \node (n51)  at (5,2){$\vee$};
 \node (n61)  at (6,2.5)  {$\wedge$};
 \node (n71)  at (5,3) {$\vee$};
 \node (a71)  at (4.5,2.75) {$0$};
 
 \draw (a51) -- (n51) 
    -- (n61) -- (n71) -- (a71);
   
 \foreach \x/ \y/ \node in{
  {5.5 / 1.75/ (n51)},
  {6.5 / 2.25/ (n61)}}
  {
   \draw[dashed] (\x,\y) -- \node;
  };
  
  \draw[rotate around={-60:(4.75,2.875)}, dashed] (4.75,2.875) ellipse (0.26 and 0.52);

  \node (to2) at (7.5,2) {\large $\Rightarrow$};
  
  \node (a62)  at (8.5,1.75) {$0$};
  \node (n52)  at (9,2)  {$\vee$};
  \node (n62)  at (10,2.5) {$\wedge$};
 
  \draw (a62) -- (n52) -- (n62) ;
  
  \draw[dashed] (n52) -- (9.5,1.75);
  \draw[dashed] (n62) -- (10.5,2.25);
  
 \end{tikzpicture}
 
  \caption{An example for step $1$ in the proof of Theorem~\ref{theorem:eval_bool}.
  In the first image we find the expression $\wedge 0$,
  hence we remove the remaining suffix.
  The expression $\vee 0$ can also be removed without changing the final truth value.
 }
 \label{fig:bool}
\end{figure}

\medskip
\noindent
{\em Step 2.} To evaluate a left caterpillar tree let $\bb A_1$ and $\bb A_2$ be two SLPs
where $\val(\bb A_1) \in \{\wedge,\vee\}^*$, $\val(\bb A_2) \in \{0,1\}^*$, and $|\val(\bb A_2)| = |\val(\bb A_1)|+1$.
Let $\varphi : \{ \wedge, \vee \}^* \to \{0,1\}^*$ be the homomorphism with $\varphi(\wedge) = 1$ and 
$\varphi(\vee) = 0$. Using binary search, we compute the largest position $i$ such that the reversed length-$i$ suffix
of $\val(\bb A_2)$ is equal to the length-$i$ prefix of $\varphi(\val(\bb A_1))$. If $i = |\val(\bb A_1)|$, then
 the value of $\val(\bb A_1) \, \val(\bb A_2)$ is the first symbol of  $\val(\bb A_2)$. Otherwise, the 
value of $\val(\bb A_1) \, \val(\bb A_2)$ is $0$ (resp., $1$) if $\val(\bb A_1)[i+1] = \wedge$ (resp., $\val(\bb A_1)[i+1] = \vee$).
\end{proof}

\begin{corollary}
If the interpretation $\mc I$ is such that $(\mc D, \wedge^{\mc I}, \vee^{\mc I})$ is a finite distributive lattice, then the $\mc I$-evaluation problem for SLP-compressed trees
can be solved in polynomial time.
\end{corollary}

\begin{proof}
By Birkhoff's representation theorem, every finite distributive lattice is isomorphic to a lattice of finite sets, 
where the join (resp., meet) operation is set union (resp., intersection).
This lattice embeds into a finite power of $(\{0,1\}, \wedge, \vee)$. 
\end{proof}

\subsubsection{Difficult arithmetical evaluation problems}
\label{sec-arithmetic-eval}

Assume that $\mc I$ is the interpretation that assigns to the symbols $+$ and $\times$ their standard
meaning over the  integers. Note that this interpretation is not polynomially bounded.
For instance, for the tree $t_n = \times^n  (2)^{n+1}$ we have $t_n^{\mc I} = 2^{n+1}$.
Hence, if a tree $t$ is given by an SLP $\bb A$, then the number of bits of $t^{\mc I}$ can be exponential in the size 
of $\bb A$. Therefore, we cannot write down the number $t^{\mc I}$ in polynomial time. The same problem arises
already for numbers that are given by arithmetic circuits (circuits over $+$ and $\times$).

 In \cite{AllenderBKM09} it was shown that the problem of computing the $k^{\text{th}}$ bit 
($k$ is given in binary notation) of the number to which a given arithmetic circuit  evaluates to belongs
to the counting hierarchy.  An arithmetic circuit can be seen as a dag that unfolds to an expression tree.
Dags correspond to TSLPs where all nonterminals have rank $0$. Vice versa, it was shown in 
\cite{HuckeLN14arxiv} that a TSLP $\bb A$ over $+$ and $\times$ can be transformed in logspace
into an arithmetic circuit that evaluates to $\val(\bb A)^{\mc I}$. This transformation holds for any semiring.
Thus, over semirings, the evaluation problems for TSLPs and circuits (i.e., dags) have the same complexity.
In particular, the problem of computing the $k^{\text{th}}$ bit of the output value of a TSLP-represented
arithmetic expression belongs to the counting hierarchy. 
Here, we show that this result even holds for 
arithmetic expressions that are given by SLPs:

\begin{theorem} \label{thm-plus-mal}
The problem of computing for a given binary encoded number $k$ and an SLP $\bb A$ over $\{+, \times\} \cup \bb Z$ 
the $k^{\text{th}}$ bit of $\val(\bb A)^{\mc I}$ belongs to the counting hierarchy.
\end{theorem}

\begin{proof}
We follow the strategy from \cite[proof of Thm.~4.1]{AllenderBKM09}. Let $\bb A$ be the input SLP for the tree $t$ and let
$N = {\mc I}(t)$.  Then $N \leq 2^{2^n}$ where $n = |\bb A|$ (this follows since the expression $t$ has size 
at most $2^n$ and the value computed by an expression of size $m$ is at most $2^m$). Let $P_n$ be the set of all prime numbers
in the range $[2, 2^{2n}]$ (note that $2^{2n} \ge \log^2 N$). Then $\prod_{p \in P_n}  p > N$. Also note that
each prime $p \in P_n$ has at most $2n$ bits in its binary representation.
We first show that the language
\begin{align*}
L = \{ (\bb A, p, j) \mid \; & \bb A \text{ is an SLP for a tree}, ~ n = |\bb A|, ~ p \in P_n, ~ 1 \leq j \leq 2n,  \\
& \text{the $j^{\text{th}}$ bit of $\val(\bb A)^{\mc I} \bmod p$ is $1$} \}
\end{align*}
belongs to the counting hierarchy. The rest of proof then follows the argument in \cite{AllenderBKM09}: 
Using the $\mathsf{DLOGTIME}$-uniform $\mathsf{TC}^0$-circuit family from~\cite{HeAlBa02}
for transforming a number from its Chinese remainder representation into its binary representation
one defines a $\mathsf{TC}^0$-circuit of size $2^{\mc O(n)}$ that has input gates 
$x(p, j)$ (where $n = |\bb A|$, $p \in P_n$,  $1 \leq j \leq 2n$). If we set $x(p, j)$ to true iff $(\bb A, p, j)  \in L$
(this means that the input gates $x(p,j)$ receive the Chinese remainder representation of $\val(\bb A)^{\mc I}$),
then the circuit outputs correctly the (exponentially many)
bits of the binary representation of $\val(\bb A)^{\mc I}$. Then, as in \cite[proof of Thm.~4.1]{AllenderBKM09},
one shows by induction on the depth of a gate that the problem whether a given gate of that circuit 
(the gate is specified by a bit string of length $\mc O(n)$)
evaluates to true is in the counting hierarchy, where the level in the counting hierarchy depends on the level 
of the gate in the circuit.\footnote{Let us explain the differences to \cite[proof of Thm.~4.1]{AllenderBKM09}:
In \cite{AllenderBKM09}, the arithmetic expression is given by a circuit instead of 
an SLP. This simplifies the proof, because if we replace in the above language $L$
the SLP $\bb A$ by a circuit, then we can decide the language $L$ in polynomial time
(we only have to evaluate a circuit modulo a prime number with polynomially many bits).
In our situation, we can only show that $L$ belongs to a certain level of the counting hierarchy.
But this suffices to prove the theorem, only the level in the counting hierarchy increases by the number
of levels in which the set $L$ sits.}

Hence we have to show that $L$ belongs to the counting hierarchy. 
Let $\bb A$ be an SLP for a tree $t$, $n = |\bb A|$, $p \in P_n$, and $1 \leq j \leq 2n$.
By Theorem~\ref{thm:main-reduction} it suffices to consider the case that $t$ is a caterpillar tree $t$; the polynomial time
Turing reduction in Theorem~\ref{thm:main-reduction} increases the level in the counting hierarchy 
by one. Also note that we use a uniform version of 
 Theorem~\ref{thm:main-reduction}, where the interpretation (addition and multiplication in $\mathbb{Z}_p$) is part of the input. This is not
 a problem, since the prime number $p$ has at most $2n$ bits, so all values that can appear only need $2n$ bits.
 
Let $m$ be the number of operators in $t$, i.e., the total number of occurrences of the symbols $+$ and $\times$ in
$\val(\bb A)$. Note that $m$ can be exponentially large in $|\bb A|$,
but its binary representation can be computed in polynomial time by Lemma~\ref{lemma:slp_folklore}
(point \ref{lemma:number-of-occ}).
We now define a matrix of numbers $x^t_{i,j} \in \mathbb{Z}_p$ ($i,j \in [1,m+1]$) such that
\[
t^{\mc I} = \sum_{i=1}^{m+1} \prod_{j=1}^{m+1} x^t_{i,j}.
\]
Moreover, we will show that given $\bb A$ and binary encoded numbers $i,j \in [1,m+1]$, the binary encoding
of $x^t_{i,j}$ (which consists of at most $2n$ bits) can be computed in polynomial time. 

We define the numbers $x^t_{i,j}$ inductively over the structure of the caterpillar tree $t$.
For the caterpillar tree $t = a$ (with $a \in \mathbb{Z}_p$) we set $x^t_{1,1}=a$.
Now assume that $t = f(a,s)$ or $t = f(s,a)$ for an operator $f \in \{+,\times\}$, a 
caterpillar tree $s$ with $m-1$ operators, and $a \in \mathbb{Z}_p$.
In the case $t = f(s,a)$ we assume that $m-1 \geq 1$; this avoids 
ambiguities in case $t = f(a,b)$ for $a,b \in \mathbb{Z}_p$.
Assume that the numbers $x_{i,j}^s$ are already defined for $i,j \in [1,m]$.
If $f = +$, then we set:
\begin{eqnarray*}
x^t_{1,1} &=& a \\
x^t_{1,i} & = & 1 \text{ for } i \in [2,m+1]\\
x^t_{i,1} & = & 1 \text{ for } i \in [2,m+1]\\
x^t_{i,j} & = & x^s_{i-1,j-1} \text{ for } i,j \in [2,m+1]
\end{eqnarray*}
We get
\[
 \sum_{i=1}^{m+1} \prod_{j=1}^{m+1} x^t_{i,j} 
    =  a + \sum_{i=2}^{m+1} \prod_{j=2}^{m+1} x^s_{i-1,j-1} 
    =  a + \sum_{i=1}^{m} \prod_{j=1}^{m} x^s_{i,j} 
    =  a + s^{\mc I} = t^{\mc I}.
\]
If $f = \times$, then we set:
\begin{eqnarray*}
x^t_{1,i} &=& 0 \text{ for } i \in [1,m+1] \\
x^t_{i,1} & = & a \text{ for } i \in [2,m+1]\\
x^t_{i,j} & = & x^s_{i-1,j-1} \text{ for } i,j \in [2,m+1]
\end{eqnarray*}
We get
\[
 \sum_{i=1}^{m+1} \prod_{j=1}^{m+1} x^t_{i,j}
   =  \sum_{i=2}^{m+1} a \cdot \prod_{j=2}^{m+1}  x^s_{i-1,j-1} 
   = a \cdot \sum_{i=1}^{m} \prod_{j=1}^{m} x^s_{i,j} 
   = a \cdot s^{\mc I} = t^{\mc I}.
\]
We now show that the binary encodings of the numbers $x^t_{i,j}$ can be computed in polynomial time (given $\bb A,i,j$).
For this let us introduce some notations: For our caterpillar tree $t = \val(\bb A)$ (which contains $m$ occurrences of operators)
and $i \in [1,m]$, $j \in [1,m+1]$ we define inductively $\text{op}(t,i) \in \{+,\times\}$ and $\text{operand}(t,j) \in \mathbb{Z}_p$ as follows:
\begin{itemize}
\item If $t = a \in \mathbb{Z}_p$, then let $\text{operand}(t,1) = a$ (note that in this case we have $m=0$, 
      hence the  $\text{op}(t,i)$ do not exist).
\item If $t = f(a,s)$ or ($t = f(s,a)$ and $m-1 \geq 1$) with $a \in \mathbb{Z}_p$, then we set
$\text{op}(t,1) = f$, $\text{op}(t,i) = \text{op}(s,i-1)$ for $i \in [2,m]$, $\text{operand}(t,1) = a$, and
$\text{operand}(t,j) = \text{operand}(s,j-1)$ for $j \in [2,m+1]$.
\end{itemize}
In other words: $\text{op}(t,i)$ is the $i^{\text{th}}$ operator in $t$, and $\text{operand}(t,j)$ is the unique argument from $\mathbb{Z}_p$
of the $j^{\text{th}}$ operator in $t$ (recall that $t$ is a caterpillar tree). 
The $m^{\text{th}}$ (and hence last) operator in $t$ has two arguments from $\mathbb{Z}_p$; its
left argument is $\text{operand}(t,m)$ and its right argument is $\text{operand}(t,m+1)$.
Using these notations, we can compute the numbers
$x_{i,j}^t$ by the following case distinction (correctness follows by a straightforward induction):
\begin{itemize}
\item $i < j$: If $\text{op}(t, i) = +$ then $x_{i,j}^t = 1$, else $x_{i,j}^t = 0$.
\item $i = j$: If $\text{op}(t, i) = +$ then $x_{i,j}^t = \text{operand}(t, j)$, else $x_{i,j}^t = 0$.
\item $i > j$: If $\text{op}(t, j) = +$ then $x_{i,j}^t = 1$, else $x_{i,j}^t = \text{operand}(t, j)$.
\end{itemize}
So, in order to compute the $x_{i,j}^t$ it suffices to compute $\text{op}(t, i)$ and $\text{operand}(t, j)$, given $\bb A, i,j$.
This is possible in polynomial time: 
The position $k$ of the $i^{\text{th}}$ operator in $t$ and  $\text{op}(t, i)$ can be computed in polynomial time using
point \ref{lemma:firstpos} of  Lemma~\ref{lemma:slp_folklore} (take $\Gamma = \{+, \times\}$). Once the position
$k$ is computed, $\text{operand}(t, i)$ can be computed in polynomial time using point (b) of Theorem~\ref{thm:navi}.

Recall that our goal is to compute a specific bit of $\val(\bb A)^{\mc I} \bmod p$, where $\bb A$ is an SLP that produces a caterpillar tree,
and $p \in [2,2^{2n}]$ is a prime, where $n = |\bb A|$. We have to show that this problem belongs to the counting hierarchy.
We have shown that 
\[
\val(\bb A)^{\mc I} = \sum_{i=1}^{m+1} \prod_{j=1}^{m+1} x^t_{i,j}.
\]
where the binary encoding of the number $x_{i,j}^t \in \mathbb{Z}_p$ can be computed in polynomial time, given $\bb A, i,j$.
We now follow again the arguments from \cite{AllenderBKM09}. It is known that the binary representation of a sum (resp., product) of $n$ many $n$-bit numbers 
can be computed in $\mathsf{DLOGTIME}$-uniform $\mathsf{TC}^0$ \cite{HeAlBa02}. The same holds for the problem of computing 
 a sum (resp., product) of $n$ many  numbers from $[0,p-1]$ modulo a given prime number $p$ with $\mc O(\log n)$ bits (it is actually much easier
 to argue that the latter problem is in $\mathsf{DLOGTIME}$-uniform $\mathsf{TC}^0$, see again \cite{HeAlBa02}).
Hence, there is a  $\mathsf{DLOGTIME}$-uniform $\mathsf{TC}^0$ circuit family $(C_m)_{m \geq 1}$, where the input of $C_m$ consists of bits
$x(i,j,k)$ ($i,j \in [1,m]$, $k \in \mc O(\log m)$) and a prime number $p$ with $\mc O(\log m)$ bits, 
such that the following holds: If $x(i,j,k)$ receives the $k^{\text{th}}$ bit of a number $x_{i,j} \in \mathbb{Z}_p$, then the 
circuit outputs $\sum_{i=1}^{m} \prod_{j=1}^{m} x_{i,j} \bmod p$. We take the circuit $C_{m+1}$, where $m \in 2^{\mc O(n)}$ (recall that
$n = |\bb A|$ and $m$ is the number of operators in $t = \val(\bb A)$). 
The input gate $x(i,j,k)$ receives the $k^{\text{th}}$ bit of the number $x^t_{i,j} \in  \mathbb{Z}_p$ defined above.
We have shown above that the bits of $x^t_{i,j}$ can be computed in polynomial time. This allows (again in the same way as in 
\cite[proof of Thm.~4.1]{AllenderBKM09}) to show that for a given gate number of $C_{m+1}$ one can compute the truth value of 
the corresponding gate within the counting hierarchy.
\end{proof}
Computing a certain bit of the output number of an arithmetic circuit 
belongs to $\mathsf{PH}^{\mathsf{PP}^{\mathsf{PP}^{\mathsf{PP}}}}$ \cite{DBLP:conf/mfcs/AllenderBD14}
(but  no matching lower bound is known).
In our situation, the level gets even higher, so we made no effort
to compute it.

We can use the technique from the  proof of Theorem~\ref{thm-plus-mal} to show the following related result. Note that a circuit (or dag) over
$\max$ and $+$ can be evaluated in polynomial time (simply by computing bottom-up the value of each gate), 
and by the reduction from \cite{HuckeLN14arxiv} the same holds for TSLP-compressed
expressions.

\begin{theorem}
The problem of evaluating SLP-compressed $(\{\max,+\} \cup \bb Z)$-trees over the integers belongs to the counting hierarchy.
\end{theorem}

\begin{proof}
The proof follows the arguments from the proof of Theorem~\ref{thm-plus-mal}. 
But since the interpretation given by $\max$ and $+$ is polynomially bounded, 
every subtree of an SLP-compressed tree evaluates to an integer that needs only polynomially many bits with respect
to the size of the SLP. Hence we do not need the Chinese remainder theorem as in the proof of Theorem~\ref{thm-plus-mal} 
and can use Theorem~\ref{thm:main-reduction} directly. 
It remains to show that the problem of evaluating  SLP-compressed $(\{\max,+\} \cup \bb Z)$-caterpillar trees
belongs to the counting hierarchy. For this we follow the same strategy as in the proof of Theorem~\ref{thm-plus-mal} and define 
numbers $x^t_{i,j}$ (where $t = \val(\bb A)$ is the input caterpillar tree) such that
$$
\val(\bb A)^{\mc I} = \max_{1 \le i \le m+1} \sum_{j=1}^{m+1} x^t_{i,j}.
$$
Since the sum of $n$ many $n$-bit numbers as well as the maximum of $n$ many $n$-bit numbers can be computed
in $\mathsf{DLOGTIME}$-uniform $\mathsf{TC}^0$ (the maximum of $n$ many $n$-bit numbers can be even computed
in $\mathsf{DLOGTIME}$-uniform $\mathsf{AC}^0$), one can argue as in the proof of Theorem~\ref{thm-plus-mal}.
\end{proof}
Let us now turn to lower bounds for the problems of evaluating SLP-compressed arithmetic expressions (max-plus or plus-times). 
For a number $c \in \ns$ consider the unary operation $+_c$ on $\ns$ with $+_c(z) = z+c$.
The evaluation of SLP-compressed $(\{\max,+_c\} \cup \ns)$-trees is possible in polynomial time analogously 
to the proof of Theorem~\ref{thm-height}.
The following theorem shows that the general case of SLP-compressed $(\{\max,+\} \cup \bb N)$-trees is more complicated.

\begin{theorem}\label{thm:max-plus}
Evaluating SLP-compressed $(\{\max,+\} \cup \bb N)$-trees is $\# \P$-hard.
\end{theorem}

\begin{proof}
	Let $\bb A, \bb B$ be two SLPs over $\{0,1\}$ with $|\val(\bb A)| = |\val(\bb B)|$. 
	We will reduce from the problem of counting the number of occurrences of $(1,1)$ 
	in the convolution $\val(\bb A) \otimes \val(\bb B) \in (\{0,1\}^2)^*$, 
	which is known to be $\# \P$-complete by \cite{DBLP:journals/iandc/Lohrey11}.
	Let $\rho: \{0,1\}^* \to \{\max,+\}^*$ be the homomorphism defined by $\rho(0) = \max, ~\rho(1) = +$.
	One can compute in polynomial time from $\bb A$ and $\bb B$ an SLP for the tree $\rho(\val(\bb A)) \, 1 \, \mathrm{rev}(\val(\bb B))$.
	The corresponding tree over $\{\max,+,0,1\}$ evaluates to one plus the number of occurrences
	of $(1,1)$ in the convolution $\val(\bb A) \otimes \val(\bb B)$. 
\end{proof}
In \cite{AllenderBKM09} it was shown that the computation of a certain bit of the output value of an arithmetic circuit (over $+$ and $\times$) is 
$\#\P$-hard. Since a circuit can be seen as a TSLP (where all nonterminals have rank $0$), which can be transformed in polynomial time into
an SLP for the same tree \cite{BuLoMa07}, also the problem of computing a certain bit of $\val(\bb A)^{\mc I}$ for a given SLP $\bb A$ is 
$\#\P$-hard. For the related problem PosSLP of deciding, whether a given arithmetic circuit computes a positive number, no non-trivial lower
bound is known. For SLPs, the corresponding problem becomes $\mathsf{PP}$-hard:

\begin{theorem}
The problem of deciding whether $\val(\bb A)^{\mc I} \geq 0$ for a given SLP
$\bb A$ over $\{+, \times\} \cup \bb Z$ is $\mathsf{PP}$-hard.
\end{theorem}

\begin{proof}
	By \cite{DBLP:journals/iandc/Lohrey11}, the following problem is $\mathsf{PP}$-complete: 
	Given SLPs $\bb A, \bb B$ over $\{0,1\}$ where $|\val(\bb A)|=|\val(\bb B)|$, and a binary encoded
	number $z$, is the number of occurrences of $(1,1)$ in the convoluted string $\val(\bb A) \otimes \val(\bb B)$ at least $z$?
	We modify the proof of Theorem \ref{thm:max-plus}. 
	Let $\bb A, \bb B$ be SLPs over $\{0,1\}$, where $N = |\val(\bb A)|=|\val(\bb B)|$. 
	Pick $n \ge 0$ such that $2^n > 2N$. 
	Let $\rho_A: \{0,1\}^* \to \{+,\times\}^*$ be the homomorphism defined by $\rho_A(0) = +, ~\rho_A(1) = \times$
	and $\rho_B: \{0,1\}^* \to \{1,2\}^*$ 
	be the homomorphism defined by $\rho_B(0) = 1, ~\rho_B(1) = 2$.
	One can compute in polynomial time from $\bb A$ and $\bb B$ an SLP for the tree 
	$\rho_A(\val(\bb A))\, (2^n)\, \rho_B(\mathrm{rev}(\val(\bb B)))$ (here $2^n$ stands for an SLP that evaluates to $2^n$).
	Let $R$ be the value of the corresponding tree. Note that $R$ is calculated by starting with the value $2^n$ 
	and applying $N$ additions or multiplications by 1 or 2.
	The number $K$ of occurrences of $(1,1)$ in the convolution $\val(\bb A) \otimes \val(\bb B)$
	corresponds to the number of multiplications by 2 in the calculation, which can be computed from $R$:
	We have
	\[
	  2^n \cdot 2^K \le R \le (2^n + 2(N-K)) \cdot 2^K \leq  (2^n + 2N) \cdot 2^K
	\]
	since $R$ is maximal if $(N-K)$ additions of 2 are followed by $K$ multiplications by 2. Since $2N < 2^n$ we obtain 
	$2^{n+K} \le  R \leq 2^{n+K}  +r$ for some $r < 2^{n+K}$. Hence, $K \geq z$, if and only if $R - 2^{n+z}\geq 0$. 
	It is straightforward to compute an SLP which evaluates to $R - 2^{n+z}$.
\end{proof}

\subsubsection{Tree automata}

(Bottom-up) tree automata (see \cite{tata97} for details) can be seen as finite algebras: 
The domain of the algebra is the set of states, and the operations 
of the algebra correspond to the transitions of the automaton. This correspondence only holds for deterministic tree
automata. On the other hand every nondeterministic tree automaton can be transformed into a deterministic one using a powerset
construction. Formally, a {\em nondeterministic (bottom-up) tree automaton} $\mc A = (Q, \mc F, \Delta,F)$ consists of a finite set of 
{\em states} $Q$, a ranked alphabet $\mc F$, a set $\Delta$ of {\em transition rules} 
of the form $f(q_1, \dots, q_n) \to q$ where $f \in \mc F_n$ and $q_1, \dots, q_n, q \in Q$, 
and a set of {\em final states} $F \subseteq Q$. A tree $t \in \mc T(\mc F)$ is {\em accepted} by $\mc A$ if $t \stackrel{*}{\to}_\Delta q$ 
for some $q \in F$ where $\to_\Delta$ is the rewriting relation defined by $\Delta$ as usual.
The uniform membership problem for tree automata asks whether a given tree automaton $\mc A$ accepts a given tree $t \in \mc T(\mc F)$. 
In \cite{loh01} it was shown that this problem is complete for the class \LogCFL,
which is the closure of the context-free languages under 
logspace reductions. 
\LogCFL is contained in \P and {\sf DSPACE}($\log^2(n)$).  
For every fixed tree automaton, the membership problem belongs to $\mathsf{NC}^1$ \cite{loh01}.
If the input tree is given by a TSLP, the uniform membership problem becomes \P-complete \cite{LoMaSS12}.
For non-linear TSLPs (where a parameter may occur several times in a right-hand side) the uniform membership problem becomes \PSPACE-complete,
and \PSPACE-hardness holds already for a fixed tree automaton  \cite{DBLP:journals/tcs/LohreyM06}. 
The same complexity bound holds for SLP-compressed trees (which in contrast to non-linear TSLPs only allow exponential compression):

\begin{theorem} \label{thm-TA}
	Given a tree automaton $\mc A$ and an SLP $\bb A$ for a tree $t \in \mc T(\mc F)$, it is \PSPACE-complete to decide whether $\mc A$ accepts $t$.
	Moreover, \PSPACE-hardness already holds for a fixed tree automaton.
\end{theorem}

\begin{proof}
	For the upper bound we use the following lemma from \cite{DBLP:journals/iandc/LohreyM13}: 
	If a function $f: \Sigma^* \to \Gamma^*$ is \PSPACE-computable and $L \subseteq \Gamma^*$ belongs to {\sf NSPACE}$(\log^k(n))$ 
	for some constant $k$, then $f^{-1}(L)$ belongs to \PSPACE. 
	Given an SLP $\bb A$ for the tree $t =\val(\bb A)$, one  
	can compute the tree $t$ by a \PSPACE-transducer by computing the symbol $t[i]$ 
	for every position $i \in \{1, \dots, |t|\}$. 
	The current position can be stored in polynomial space and every query can be performed in polynomial time. 
	As remarked above the uniform membership problem for explicitly given trees can be solved in {\sf DSPACE}$(\log^2(n))$.
	
	For the lower bound we use a fixed regular language $L \subseteq (\{0,1\}^2)^*$ from \cite{DBLP:journals/iandc/Lohrey11} 
	such that the following problem is \PSPACE-complete: 
	Given two SLPs $\bb A$ and $\bb B$ over $\{0,1\}$ with $|\val(\bb A)| = |\val(\bb B)|$, 
	is $\val(\bb A) \otimes \val(\bb B) \in L$?
	
	Let $\mc A = (Q,\{0,1\}^2,\Delta,q_0,F)$ be a finite word automaton for $L$. 
	Let $\bb A, \bb B$ be two SLPs over $\{0,1\}$ with $|\val(\bb A)| = |\val(\bb B)|$ 
	and let $\bb T$ be an SLP for the comb tree $t(u,v)$ where $u = \mathrm{rev}(\val(\bb A))$ 
	and $v = \mathrm{rev}(\val(\bb B))$. We transform $\mc A$ into a tree automaton $\mc A_T$ over $\{f_0,f_1,0,1,\$\}$
	with the state set $Q  \uplus \{p_0,p_1\}$, the set of final states $F$ and the following transitions:
	\begin{align*}
		\$ &\to q_0, \\
		i &\to p_i, \quad \text{for } i \in \{0,1\}, \\
		f_i(q,p_j) &\to q', \quad \text{for } (q,(i,j),q') \in \Delta
	\end{align*}
	The automaton $\mc A$ accepts the convolution $\val(\bb A) \otimes \val(\bb B)$ if and only if the tree automaton $\mc A_T$ accepts $t(u,v)$.
\end{proof}
The \PSPACE-hardness result in Theorem~\ref{thm-TA} can also be interpreted as follows: 
There exists a fixed finite algebra for which the evaluation 
problem for  SLP-compressed trees is $\PSPACE$-complete. 
This is a bit surprising if we compare the situation with dags or TSLP-compressed trees. For these,
membership for tree  automata is still doable in polynomial time \cite{LoMaSS12}, 
whereas the evaluation problem of arithmetic expressions (in the sense of computing a certain bit of the 
output number) belongs to the counting hierarchy and is $\# \P$-hard. 
In contrast, for SLP-compressed trees, the evaluation problem for  finite algebras (i.e., tree automata)
is harder than the evaluation problem for arithmetic expressions ($\PSPACE$ versus the counting hierarchy).

\section{Further research}

We conjecture that in practice, grammar-based tree compression based on  SLPs leads to faster compression
and better compression ratios compared to grammar-based tree compression based on TSLPs, and we plan to substantiate this
conjecture with experiments on real tree data. The theoretical results from Section~\ref{section:slpvstslp} indicate that
SLPs may achieve better compression ratios than TSLPs. 
Moreover, grammar-based string compression can be implemented without pointer structures, whereas
all grammar-based tree compressors (that construct TSLPs) we are aware of  
work with pointer structures for trees, and a string-encoded tree (e.g. an XML document) must be first transformed into a pointer
structure. Moreover, we believe that SLPs can be encoded more succinctly
than TSLPs (for instance, we do not have to store the ranks of nonterminals).

\bibliographystyle{plain}


\begin{thebibliography}{10}

\bibitem{Akutsu10}
T.~Akutsu.
\newblock A bisection algorithm for grammar-based compression of ordered trees.
\newblock {\em Inf.~Process.~Lett.}, 110(18-19):815--820, 2010.

\bibitem{DBLP:conf/mfcs/AllenderBD14}
E.~Allender, N.~Balaji, and S.~Datta.
\newblock Low-depth uniform threshold circuits and the bit-complexity of
  straight line programs.
\newblock In {\em Proceedings of the 39th International Symposium of
  Mathematical Foundations of Computer Science, {MFCS} 2014, Part {II}}, volume
  8635 of {\em Lecture Notes in Computer Science}, pages 13--24. Springer,
  2014.

\bibitem{AllenderBKM09}
E.~Allender, P.~B{\"u}rgisser, J.~Kjeldgaard-Pedersen, and P.~Bro
  Miltersen.
\newblock On the complexity of numerical analysis.
\newblock {\em SIAM Journal on Computing}, 38(5):1987--2006, 2009.

\bibitem{DBLP:journals/eatcs/AllenderW90}
E.~Allender and K.~W.~Wagner.
\newblock Counting hierarchies: Polynomial time and constant.
\newblock {\em Bulletin of the {EATCS}}, 40:182--194, 1990.

\bibitem{DBLP:conf/ifipTCS/BertoniCR08}
A.~Bertoni, C.~Choffrut, and R.~Radicioni.
\newblock Literal shuffle of compressed words.
\newblock In {\em Proceedings of the 5th {IFIP} International Conference on
  Theoretical Computer Science, {TCS} 2008, {IFIP} 20th World Computer
  Congress}, volume 273 of {\em
  {IFIP}}, pages 87--100. Springer, 2008.
  
\bibitem{BenoitDMRRR05}
D.~Benoit, E.~D. Demaine, J.~I. Munro, R.~Raman, V.~Raman, and
  S.~S. Rao.
\newblock Representing trees of higher degree.
\newblock {\em Algorithmica}, 43(4):275--292, 2005.
  
\bibitem{BilleGLW13}
P.~Bille, I.~L. G{\o}rtz, G.~M. Landau, and O.~Weimann.
\newblock Tree compression with top trees.
\newblock {\em Information and Computation}, 243: 166--177, 2015.

\bibitem{BLRSSW15}
P.~Bille, G.~M. Landau, R.~Raman, K.~Sadakane, S.~R. Satti, and O.~Weimann.
\newblock Random access to grammar-compressed strings and trees.
\newblock {\em SIAM Journal on Computing}, 44(3)513--539:, 2015.


\bibitem{MLMN13}
M.~Bousquet-M{\'e}lou, M.~Lohrey, S.~Maneth, and E.~Noeth.
\newblock {XML} compression via {DAG}s.
\newblock {\em Theory of Computing Systems}, 2014.

\bibitem{BuLoMa07}
G.~Busatto, M.~Lohrey, and S.~Maneth.
\newblock Efficient memory representation of {XML} document trees.
\newblock {\em Information Systems}, 33(4--5):456--474, 2008.

\bibitem{Bus87}
Samuel~R. Buss.
\newblock The {Boolean} formula value problem is in {ALOGTIME}.
\newblock In {\em Proc.~STOC 87}, pages 123--131. ACM Press, 1987.

\bibitem{DBLP:journals/jcss/CaussinusMTV98}
H.~Caussinus, P.~McKenzie, D.~Th{\'{e}}rien, and H.
  Vollmer.
\newblock Nondeterministic \emph{NC}\({}^{\mbox{1}}\) computation.
\newblock {\em Journal of Computer and System Sciences}, 57(2):200--212, 1998.

\bibitem{CLLLPPSS05}
M.~Charikar, E.~Lehman, A.~Lehman, D.~Liu, R.~Panigrahy, M.~Prabhakaran,
  A.~Sahai, and A.~Shelat.
\newblock The smallest grammar problem.
\newblock {\em IEEE Transactions on Information Theory}, 51(7):2554--2576,
  2005.

\bibitem{tata97}
H.~Comon, M.~Dauchet, R.~Gilleron, F.~Jacquemard, D.~Lugiez, C.~L{\"o}ding,
  S.~Tison, and M.~Tommasi.
\newblock Tree automata techniques and applications.
\newblock {\tt http://tata.gforge.inria.fr/}.

\bibitem{DBLP:conf/lata/EsparzaLS14}
J.~Esparza, M.~Luttenberger, and M.~Schlund.
\newblock A brief history of strahler numbers.
\newblock In {\em Proceedings of the 8th International Conference on Language
  and Automata Theory and Applications, {LATA} 2014}, volume 8370 of {\em
  Lecture Notes in Computer Science}, pages 1--13. Springer, 2014.

\bibitem{FerraginaLMM09}
P.~Ferragina, F.~Luccio, G.~Manzini, and S.~Muthukrishnan.
\newblock Compressing and indexing labeled trees, with applications.
\newblock {\em J.~ACM}, 57(1), 2009.


\bibitem{DBLP:journals/tcs/FlajoletRV79}
P.~Flajolet, J.{-}C.~Raoult, and J.~Vuillemin.
\newblock The number of registers required for evaluating arithmetic
  expressions.
\newblock {\em Theoretical Computer Science}, 9:99--125, 1979.

\bibitem{HuckeLN14arxiv}
M.~Ganardi, D.~Hucke, A.~J{\.{e}}z, M.~Lohrey, and E.~Noeth.
\newblock Constructing small tree grammars and small circuits for formulas.
\newblock Technical report, arXiv.org, 2014.
\newblock \url{http://arxiv.org/abs/1407.4286}.

\bibitem{Hag00}
C.~Hagenah.
\newblock {\em Gleichungen mit regul{\"a}ren Randbedingungen {\"u}ber freien
  Gruppen}.
\newblock PhD thesis, University of {Stuttgart}, Institut f{\"u}r Informatik,
  2000.

\bibitem{HeAlBa02}
W.~Hesse, E.~Allender, and D.~A.~Mix Barrington.
\newblock Uniform constant-depth threshold circuits for division and iterated
  multiplication.
\newblock {\em Journal of Computer and System Sciences}, 65:695--716, 2002.

\bibitem{Hubschle-Schneider15}
L.~H{\"{u}}bschle{-}Schneider and R.~Raman.
\newblock Tree compression with top trees revisited.
\newblock In {\em Proc.~SEA 2015}, volume 9125 of {\em LNCS}, pages 15--27.
  Springer, 2015.

\bibitem{HuckeLN14}
D.~Hucke, M.~Lohrey, and E.~Noeth.
\newblock Constructing small tree grammars and small circuits for formulas.
\newblock In {\em Proceedings of the 34th International Conference on
  Foundation of Software Technology and Theoretical Computer Science, {FSTTCS}
  2014}, volume~29 of {\em LIPIcs}, pages 457--468. Schloss Dagstuhl -
  Leibniz-Zentrum f\"ur Informatik, 2014.

\bibitem{Jacobson89}
G.~Jacobson.
\newblock Space-efficient static trees and graphs.
\newblock In {\em Proc.~FOCS 1989}, pages 549--554. {IEEE} Computer Society,
  1989.
  
\bibitem{JanssonSS12}
J.~Jansson, K.~Sadakane, and W-K. Sung.
\newblock Ultra-succinct representation of ordered trees with applications.
\newblock {\em Journal of Computer and System Sciences}, 78(2):619--631, 2012.
  
\bibitem{Jez13approx}
A.~J{\.{e}}z.
\newblock Approximation of grammar-based compression via recompression.
\newblock In {\em Proceedings of the 24th Annual Symposium on Combinatorial
  Pattern Matching, CPM 2013}, volume 7922 of {\em Lecture Notes in Computer
  Science}, pages 165--176. Springer, 2013.

\bibitem{JezLo14approx}
A.~J{\.{e}}z and M.~Lohrey.
\newblock Approximation of smallest linear tree grammars.
\newblock In {\em Procedings of 32nd International Symposium on Theoretical Aspects of Computer Science, STACS 2014}, volume~25 of {\em LIPIcs}, pages
  445--457. Schloss Dagstuhl - Leibniz-Zentrum f\"ur Informatik, 2014.

\bibitem{KobayashiMS12}
N.~Kobayashi, K.~Matsuda, and A.~Shinohara.
\newblock Functional programs as compressed data.
\newblock In {\em Proc.~PEPM 2012}, pages 121--130. ACM Press, 2012.

\bibitem{loh01}
M.~Lohrey.
\newblock On the parallel complexity of tree automata.
\newblock In  {\em Proceedings of the 12th
  International Conference on Rewriting Techniques and Applications, {RTA}
  2001}, volume 2051 of {\em LNCS}, pages 201--215. Springer, 2001.

\bibitem{DBLP:journals/iandc/Lohrey11}
M.~Lohrey.
\newblock Leaf languages and string compression.
\newblock {\em Information and Computation}, 209(6):951--965, 2011.

\bibitem{LoCWP}
M.~Lohrey.
\newblock {\em The Compressed Word Problem for Groups}.
\newblock Springer, 2014.

\bibitem{Loh15}
M.~Lohrey.
\newblock Grammar-based tree compression.
\newblock In {\em Proc.~DLT 2015}, volume 9168 of {\em LNCS}, pages 46--57.
  Springer, 2015.


\bibitem{DBLP:journals/tcs/LohreyM06}
M.~Lohrey and S.~Maneth.
\newblock The complexity of tree automata and {XPath} on grammar-compressed
  trees.
\newblock {\em Theoretical Computer Science}, 363(2):196--210, 2006.

\bibitem{LohreyMM13}
M.~Lohrey, S.~Maneth, and R.~Mennicke.
\newblock {XML} tree structure compression using {RePair}.
\newblock {\em Information Systems}, 38(8):1150--1167, 2013.

\bibitem{LoMaSS12}
M.~Lohrey, S.~Maneth, and M.~Schmidt-Schau{\ss}.
\newblock Parameter reduction and automata evaluation for grammar-compressed
  trees.
\newblock {\em Journal of Computer and System Sciences}, 78(5):1651--1669,
  2012.

  
\bibitem{DBLP:journals/iandc/LohreyM13}
M.~Lohrey and C.~Mathissen.
\newblock Isomorphism of regular trees and words.
\newblock {\em Information and Computation}, 224:71--105, 2013.

\bibitem{MunroR01}
J.~I. Munro and V.~Raman.
\newblock Succinct representation of balanced parentheses and static trees.
\newblock {\em SIAM J.~Comput.}, 31(3):762--776, 2001.

\bibitem{NaOP14}
G.~Navarro, A.~Ord{\'{o}}{\~{n}}ez Pereira.
\newblock Faster compressed suffix trees for repetitive text collections.
\newblock In {\em Proc.~SEA 2014}, volume 8504 of {\em LNCS}, pages 424--435.
  Springer, 2014.
  
\bibitem{Ryt03}
W.~Rytter.
\newblock Application of {Lempel}-{Ziv} factorization to the approximation of
  grammar-based compression.
\newblock {\em Theoretical Computer Science}, 302(1--3):211--222, 2003.

\bibitem{Sakamoto05}
H.~Sakamoto.
\newblock A fully linear-time approximation algorithm for grammar-based
  compression.
\newblock {\em Journal of Discrete Algorithms}, 3(2-4):416--430, 2005.

\bibitem{DBLP:journals/corr/abs-1302-6336}
M.~Schmidt{-}Schau{\ss}.
\newblock Linear compressed pattern matching for polynomial rewriting (extended
  abstract).
\newblock In {\em Proceedings of the 7th International Workshop on Computing
  with Terms and Graphs, {TERMGRAPH} 2013}, volume 110 of {\em {EPTCS}}, pages
  29--40, 2013.

\bibitem{To91}
S.~Toda.
\newblock {PP} is as hard as the polynomial-time hierarchy.
\newblock {\em SIAM J. Comput.}, 20(5):865--877, 1991.

\bibitem{Vol99}
H.~Vollmer.
\newblock {\em Introduction to Circuit Complexity}.
\newblock Texts in Theoretical Computer Science. Springer, 1999.

\end{thebibliography}

\end{document}